\documentclass[twocolumn]{bmcart}% uncomment this for twocolumn layout and comment line below
\usepackage[utf8]{inputenc} %unicode support
\usepackage{graphicx}
\usepackage{amsmath, amssymb, amsthm}
\usepackage{bm}
\usepackage{booktabs}
\usepackage{algorithm}
\usepackage{algorithmic}
\usepackage{enumitem}
\usepackage{appendix}

\newlist{abbrv}{itemize}{1}
\setlist[abbrv,1]{label=,labelwidth=1in,align=parleft,itemsep=0.1\baselineskip,leftmargin=!}

\startlocaldefs
\newtheorem{thm}{Theorem}
\newtheorem{cor}{Corollary}

\newtheorem{defn}{Definition}

\endlocaldefs

\begin{document}

%%% Start of article front matter
\begin{frontmatter}

\begin{fmbox}
\dochead{Research}

\title{Robust Distributed Cooperative RSS-based Localization for Directed Graphs in Mixed LoS/NLoS Environments}

\author[
   addressref={aff1},                   % id's of addresses, e.g. {aff1,aff2}
   %corref={aff1},                       % id of corresponding address, if any
   email={luca.carlino@unisalento.it}   % email address
]{\inits{}\fnm{Luca} \snm{Carlino}}
\author[
   addressref={aff2},
   corref={aff2},                       % id of corresponding address, if any
   noteref={n1},                        % id's of article notes, if any
   email={djin@spg.tu-darmstadt.de}
]{\inits{DJ}\fnm{Di} \snm{Jin}}
\author[
   addressref={aff2},
   email={muma@spg.tu-darmstadt.de}
]{\inits{MM}\fnm{Michael} \snm{Muma}}
\author[
   addressref={aff2},
   email={zoubir@spg.tu-darmstadt.de}
]{\inits{AMZ}\fnm{Abdelhak M.} \snm{Zoubir}}

%%%%%%%%%%%%%%%%%%%%%%%%%%%%%%%%%%%%%%%%%%%%%%
%%                                          %%
%% Enter the authors' addresses here        %%
%%                                          %%
%% Repeat \address commands as much as      %%
%% required.                                %%
%%                                          %%
%%%%%%%%%%%%%%%%%%%%%%%%%%%%%%%%%%%%%%%%%%%%%%

\address[id=aff1]{%                           % unique id
  \orgname{Department of Innovation Engineering, University of Salento}, % university, etc
  %\street{},                     %
  %\postcode{}                                % post or zip code
  \city{Lecce},                              % city
  \cny{Italy}                                    % country
}
\address[id=aff2]{%
  \orgname{Department of Electrical Engineering and Information Technology, Technische Universit$\rm{\ddot{a}}$t Darmstadt},
  %\street{},
  %\postcode{}
  \city{Darmstadt},
  \cny{Germany}
}

%%%%%%%%%%%%%%%%%%%%%%%%%%%%%%%%%%%%%%%%%%%%%%
%%                                          %%
%% Enter short notes here                   %%
%%                                          %%
%% Short notes will be after addresses      %%
%% on first page.                           %%
%%                                          %%
%%%%%%%%%%%%%%%%%%%%%%%%%%%%%%%%%%%%%%%%%%%%%%

\begin{artnotes}
%\note{Sample of title note}     % note to the article
\note[id=n1]{Corresponding author} % note, connected to author
\end{artnotes}

% \end{fmbox}% comment this for two column layout

%%%%%%%%%%%%%%%%%%%%%%%%%%%%%%%%%%%%%%%%%%%%%%
%%                                          %%
%% The Abstract begins here                 %%
%%                                          %%
%% Please refer to the Instructions for     %%
%% authors on http://www.biomedcentral.com  %%
%% and include the section headings         %%
%% accordingly for your article type.       %%
%%                                          %%
%%%%%%%%%%%%%%%%%%%%%%%%%%%%%%%%%%%%%%%%%%%%%%

\begin{abstractbox}

\begin{abstract} % abstract
The accurate and low-cost localization of sensors using a wireless sensor network is critically required in a wide range of today's applications. We propose a novel, robust maximum likelihood-type method for distributed cooperative received signal strength-based localization in wireless sensor networks. To cope with mixed LoS/NLoS conditions, we model the measurements using a two-component Gaussian mixture model. The relevant channel parameters, including the reference path loss, the path loss exponent and the variance of the measurement error, for both LoS and NLoS conditions, are assumed to be unknown deterministic parameters and are adaptively estimated. Unlike existing algorithms, the proposed method naturally takes into account the (possible) asymmetry of links between nodes. The proposed approach has a communication overhead upper-bounded by a quadratic function of the number of nodes and computational complexity scaling linearly with it. The convergence of the proposed method is guaranteed for compatible network graphs and compatibility can be tested a priori by restating the problem as a graph coloring problem. Simulation results, carried out in comparison to a centralized benchmark algorithm, demonstrate the good overall performance and high robustness in mixed LoS/NLoS environments.
\end{abstract}

%%%%%%%%%%%%%%%%%%%%%%%%%%%%%%%%%%%%%%%%%%%%%%
%%                                          %%
%% The keywords begin here                  %%
%%                                          %%
%% Put each keyword in separate \kwd{}.     %%
%%                                          %%
%%%%%%%%%%%%%%%%%%%%%%%%%%%%%%%%%%%%%%%%%%%%%%

\begin{keyword}
\kwd{cooperative localization}
\kwd{received signal strength (RSS)}
\kwd{maximum likelihood estimation}
\kwd{wireless sensor network (WSN)}
\end{keyword}

\end{abstractbox}
\end{fmbox}% uncomment this for twcolumn layout

\end{frontmatter}

%%%%%%%%%%%%%%%%
%% Background %%
%%
\section{Introduction}
The wide spread of telecommunication systems has led to the pervasiveness of radiofrequency (RF) signals in almost every environment of daily life. Knowledge of the location of mobile devices is required or beneficial in many applications \cite{Gustafsson_SPM}, and numerous localization techniques have been proposed over the years~\cite{Gustafsson_SPM, Patwari_SPM, Gezici_survey, Trevisan2013}. Techniques based on the received signal strength (RSS) are the preferred option when low cost, simplicity and technology obliviousness are the main requirements. In some standards, e.g.~IEEE 802.15.4, an RSS indicator (RSSI) is encoded directly into the protocol stack~\cite{CR_radio}. In addition, RSS is readily available from any radio interface through a simple energy detector and can be modeled by the well-known path loss model~\cite{Rappaport} regardless of the particular communication scheme. Based on that, RSS can be exploited to implement ``opportunistic'' localization for different wireless technologies, e.g.~Wi-Fi~\cite{localizing_wifi}, FM radio~\cite{FM_radio} or cellular networks~\cite{loc_cellular}. In the context of wireless sensor networks (WSNs), nodes with known positions (\emph{anchors}) can be used to localize the nodes with unknown positions (\emph{agents}). Generally speaking, localization algorithms can be classified according to three important categories.

\emph{Centralized vs.~Distributed.} Centralized algorithms, e.g.~\cite{rss_em_centralized, Toma_ICASSP14, BCR-TSP-2015, CAMSAP2016, FUSION2016}, require a data fusion center that carries out the computation after collecting information from the nodes, while distributed algorithms, such as~\cite{tomic_rssaoa_d, rssi_distr_agri}, rely on self-localization and the computation is spread throughout the network. Centralized algorithms are likely to provide more accurate estimates, but they suffer from scalability problems, especially for large-scale WSNs, while distributed algorithms have the advantage of being scalable and more robust to node failures~\cite{Handbook}.

\emph{Cooperative vs.~Non-cooperative.} In a non-cooperative algorithm, e.g.~\cite{Yue_noncoop}, each agent receives information only from anchors. For all agents to obtain sufficient information to perform localization, non-cooperative algorithms necessitate either long range (and high-power) anchor transmission or a high-density of anchors~\cite{Handbook}. In cooperative algorithms, such as~\cite{Di_quantized, tomic_TVT}, inter-agent communication removes the need for all agents to be in range of one (or more) anchors~\cite{Handbook}.

\emph{Bayesian vs.~Non-Bayesian.} In non-Bayesian algorithms, e.g.~Expectation-Maximization (EM) \cite{YinEM}, and its variant, Expectation-Conditional Maximization (ECM)~\cite{YinECM}, the unknown positions are treated as deterministic, while in Bayesian algorithms, e.g.~Nonparametric Belief Propagation (NBP)~\cite{nbploc}, Sum-Product Algorithm over Wireless Networks (SPAWN)\\~\cite{SPAWN,SPAWN-FFT} and its variant Sigma-Point SPAWN~\cite{SPAWN-SP}, the unknown positions are assumed to be random variables with a known prior distribution.

% which can only be reduced for specific anchors configurations and only if the calibration procedure is sufficiently accurate, otherwise the noise is nonlinearly amplified.

Many existing works on localization using RSS measures, such as \cite{Patwari2003, Ouyang2010}, are based on the assumption that the classical path-loss propagation model is perfectly known, mostly via a calibration process. However, this assumption is impractical for two reasons. Firstly, conducting the calibration process requires intensive human assistance, which may be not affordable, or may even be impossible in some inaccessible areas. Secondly, the channel characteristics vary due to multipath (fading), and non-negligible modifications occur also due to mid-to-long term changes in the environment, leading to non-stationary channel parameters~\cite{infocom06}. This implies that the calibration must be performed continuously~\cite{infocom06, CR_ISWPC} because otherwise the resulting mismatch between design assumptions and actual operating conditions leads to severe performance degradation. These facts highlight the need for algorithms that \emph{adaptively} estimate the environment and the locations.
A further difficulty is due to the existence of Non-Line-of-Sight (NLoS) propagation in practical localization environments. Among various works handling the NLoS effect, a majority of them have treated the NLoS meaures as outliers and tried to neglect or mitigate their effect, including the Maximum Likelihood (ML)-based approach \cite{Riba2004,Gezici2004}, the Weighted Least-Squares (WLS) estimator \cite{Caffery1998, Gezici2004}, the constrained localization techniques \cite{Venkatesh2007,Wang2003}, robust estimators \cite{Casas2006,Sun2004} and the method based on virtual stations \cite{Liu2017}. In contrast to these works, several approaches, including \cite{Yin2013, YinEM, YinECM}, have proposed specific probabilistic models for the NLoS measures, therewith exploiting the NLoS measures for the localization purpose. In the light of these considerations, our aim is to develop an RSS-based, cooperative localization framework that works in mixed LoS/NLoS environments, requires no knowledge on parameters of the propagation model, and can be realized in a distributed manner.

\begin{table*}[tbp]
\centering
\begin{tabular}{l *{7}{c}}
\toprule
\textbf{\footnotesize{Paper}}			& \textbf{\footnotesize{RSS-based}} & \textbf{\footnotesize{Cooperative}} & \textbf{\footnotesize{Calibration-Free}}  & \textbf{\footnotesize{LoS/NLoS}}        & \textbf{\footnotesize{Distributed}}
\\ \midrule
NBP~\cite{nbploc}	& No & Yes & N/A & No  & Yes\\
SPAWN~\cite{SPAWN}  & No & Yes & N/A & No  & Yes\\
% 
% Abouzar~\emph{et al.}~\cite{rssi_distr_agri} & No 							& Yes 							      & /                                    & No(?) \\
% 
% \cite{Savarese2001,Savarese2002,Savvides2001}& No							& Yes								  & /									 & No  & Yes\\
% 
D-ECM~\cite{YinECM}  & No & Yes & N/A & Yes & Yes\\
\midrule
SDP~\cite{tomic_TVT} & Yes	& Yes & Yes	& No & No \\
EM~\cite{YinEM}	     & Yes & No   & Yes	& Yes & N/A\\
Proposed work		& Yes & Yes	& Yes & Yes  & Yes\\
\bottomrule
\end{tabular}
\caption{Succinct characterization of the related works and the proposed work. The entry ``$N/A$'' stands for ``Not applicable''. The terms ``RSS-based'', ``Cooperative'', ``Calibration-Free'', ``LoS/NLoS'', and ``Distributed'' are defined and discussed in Section 1.}
\label{tab:papers}
\end{table*}

The following distinction is made: only algorithms that directly use RSS measures as inputs are considered \emph{RSS-based} in the strict sense, while algorithms such as NBP~\cite{nbploc}, SPAWN~\cite{SPAWN}, distributed-ECM (D-ECM)~\cite{YinECM} and their variants are here called \emph{``pseudo-RSS-based''},  because they use range estimates as inputs. Generally speaking, among these two options, RSS-based location estimators are preferred for the following reasons. Firstly, inferring the range estimates from the RSS measures usually requires knowledge on the (estimated) propagation model parameters. The assumption of a priori known parameters violates the calibration-free requirement. Secondly, even with perfectly known model parameters, there exists no efficient estimator for estimating the ranges from the RSS measures, as proven in~\cite{chitte}. Thirdly, dropping the idealistic assumption of known channel parameters and using their estimates introduces an irremovable large bias, as demonstrated in~\cite{Coluccia_bias}.
% , which can only be reduced for specific anchors configurations and only if the calibration procedure is sufficiently accurate, otherwise the noise is nonlinearly amplified.
Based on these considerations, pseudo-RSS-based approaches do not meet the requirements in this work. Furthermore,  Bayesian approaches, including NBP~\cite{nbploc} and SPAWN~\cite{SPAWN,SPAWN-FFT}, do not consider mixed LoS/NLoS environments. 
% The D-ECM estimator, proposed by Yin \emph{et al.} in \cite{YinECM}, considers a mixed LoS/NLoS propagation environment, but requires range estimates as input. Although the range estimates can be inferred from the RSS measures, this procedure requires the propagation model parameters, therewith violating the calibration-free requirement. 
% The assumption of known model parameters could be justified by adopting certain self-calibration procedure, for instance~\cite{Mao_calibration}. However, this self-calibration approach is actually a \emph{centralized} procedure, which: (i) does not cope with LoS/NLoS links; (ii) assumes that one parameter of the path loss model, the reference power, is actually known. 
As one representative RSS-based cooperative localization algorithm, Tomic \emph{et al.}'s semidefinite-programming (SDP) estimator in \cite{tomic_TVT} requires no knowledge on the propagation model, but it does not apply to and cannot be readily extended to a mixed LoS/NLoS environment. To the best of our knowledge, the existing works on RSS-based calibration-free localization in a mixed LoS/NLoS environment, is rather limited. Yin \emph{et al.} have proposed an EM-based estimator in \cite{YinEM}, but only for the single agent case. In this paper, we consider a multi-agent case and aim to develop a location estimator that is RSS-based, cooperative, calibration-free and works in a mixed LoS/NLoS environment. To capture the mixed LoS/NLoS propagation conditions, we adopt the mode-dependent propagation model in \cite{YinEM}. The key difference between this work and \cite{YinEM} lies in whether the localization environment is cooperative or not. More precisely, \cite{YinEM} is concerned with the conventional single agent localization while this work studies cooperative localization in case of multi-agent. Furthermore, we develop a distributed algorithm, where model parameters and positions are updated locally by treating the estimated agents as anchors, inspired by the works in \cite{Savarese2001,Savarese2002,Savvides2001}. A succinct characterization of the proposed work and its related works is listed in Tab.~\ref{tab:papers}.

{\it Original Contributions:} We address the problem of RSS-based cooperative localization in a mixed LoS/NLoS propagation environment, requiring no calibration. To characterize such a mixed LoS/NLoS environment, we assume a mode-dependent propagation model with unknown parameters. We derive and analyze a robust, calibration-free, RSS-based distributed cooperative algorithm, based on the ML framework, which is capable of coping with mixed LoS/NLoS conditions. Simulation results, carried out in comparison with a centralized ML algorithm that serves as a benchmark, show that the proposed approach has good overall performance. Moreover, it adaptively estimates the channel parameters, has acceptable communication overhead and computation costs, thus satisfying the major requirements of a practically viable localization algorithm. The convergence analysis of the proposed algorithm is conducted by restating the problem as a graph coloring problem. In particular, we formulate a graph compatibility test, and show that for compatible network structures, the convergence is guaranteed. Unlike existing algorithms, the proposed method naturally takes into account the (possible) asymmetry of links between nodes.

The paper is organized as follows. Section 2 formulates the problem and details the algorithms. Section 3 discusses convergence. Section 4 presents simulations results, while Section 5 concludes the paper. Finally, Appendix A and Appendix B contain some analytical derivations which would otherwise burden the reading of the paper.

\section{Methods/Experimental}
\subsection{Problem Formulation}
Consider\footnote{Throughout the paper, vectors and matrices will be denoted in bold, $\| \bm{v} \|$ denotes the Euclidean norm of vector $\bm{v}$, $| \mathcal{A} |$ denotes the cardinality of set $\mathcal{A}$. We denote by $\mathbb{E}[X]$ and $\mathrm{Var}[X]$ the statistical expectation and variance, respectively, of random variable $X$. Finally, $\mathbb{B} = \{0,1\}$ is the Boolean set.} a directed graph with $N_a$ anchor nodes and $N_u$ agent nodes, for a total of $N = N_a + N_u$ nodes. In a $2$-dimensional ($2$D) scenario, we denote the position of node $i$ by $\bm{x}_i = [x_i \ y_i]^\top \in \mathbb{R}^{2 \times 1}$, where $ ^\top$ denotes transpose. Between two distinct nodes $i$ and $j$, the binary variable $o_{j \rightarrow i}$ indicates if a measure, onto direction $j \rightarrow i$, is observed ($o_{j \rightarrow i} = 1$) or not ($o_{j \rightarrow i} = 0$). In the case $i=j$, since a node does not self-measure, we have $o_{i \rightarrow i}=0$. This allows us to define the observation matrix $\bm{\mathcal{O}} \in \mathbb{B}^{N \times N}$ with elements $o_{i,j} \triangleq o_{i \rightarrow j}$ as above. The aforementioned directed graph has connection matrix $\bm{\mathcal{O}}$. It is important to remark that, for a directed graph, $\bm{\mathcal{O}}$ is not necessarily symmetric; physically, this models possible channel anisotropies, miss-detections and, more generally, link failures. Let $m_{j \rightarrow i}$ be a binary variable, which denotes if the link $j \rightarrow i$ is LoS ($m_{j \rightarrow i} = 1$) or NLoS ($m_{j \rightarrow i} = 0$). Due to physical reasons, $m_{j \rightarrow i} = m_{i \rightarrow j}$. We define the LoS/NLoS Matrix\footnote{The values on the main diagonal are arbitrary. Here we choose $m_{i \rightarrow i} = 1$.} $\bm{L} \in \mathbb{B}^{N \times N}$ of elements $l_{i,j} \triangleq m_{i \rightarrow j}$, and we observe that, since $m_{j \rightarrow i} = m_{i \rightarrow j}$, the matrix is symmetric, i.e.,  $\bm{L}^\top = \bm{L}$. We stress that this symmetry is preserved regardless of $\bm{\mathcal{O}}$, as it derives from physical reasons only. Let $\Gamma(i)$ be the \emph{(open) neighborhood} of node $i$, i.e., the set of all nodes from which node $i$ receives observables (RSS measures), formally: $\Gamma(i) \triangleq \{ j \neq i : o_{j \rightarrow i}=1 \}$. We define $\Gamma_a(i)$ as the anchor-neighborhood of node $i$. We also define $\Gamma_u(i)$ as the agent-neighborhood of node $i$, i.e., the subset of $\Gamma(i)$ which contains only agent nodes as neighbors of node $i$. In general, $\Gamma(i) = \Gamma_a(i) \cup \Gamma_u(i)$.

\subsection{Data Model}
In the sequel, we will assume that all nodes are stationary and that the observation time-window is sufficiently short in order to neglect correlation in the shadowing terms. In practice, such a model simplification allows for a more analytical treatment of the localization problem and has also been used, for example, in ~\cite{BCR-TSP-2015, Toma_ICASSP14}.
Following the path loss model and the data models present in the literature~\cite{Handbook, YinEM}, and denoting by $K$ the number of samples collected on each link over a predetermined time window, we model the received power at time index $k$ for anchor-agent links as
\begin{equation}
\label{eq:o1}
r^{(m)}_{a \rightarrow i}(k) = \begin{cases} p_{0_{\mbox{\tiny LOS}}} \hspace{-1mm} - 10  \alpha_{\mbox{\tiny LOS}} \log_{10} \| \bm{x}_a - \bm{x}_i \| + w_{a \rightarrow i}(k), 
\\ \quad \quad \text{if} \ m_{a \rightarrow i}=1; \\
p_{0_{\mbox{\tiny NLOS}}} \hspace{-2mm} - \hspace{-1mm} 10 \alpha_{\mbox{\tiny NLOS}} \log_{10} \| \bm{x}_a - \bm{x}_i \| + v_{a \rightarrow i}(k),\\
\quad \quad \text{if} \ m_{a \rightarrow i}=0,
\end{cases}
\end{equation}
while, for the agent-agent link,
\begin{equation}
\label{eq:o2}
r^{(m)}_{u \rightarrow i}(k)\hspace{-1mm} = \hspace{-1mm}\begin{cases} p_{0_{\mbox{\tiny LOS}}} \hspace{-1mm} - \hspace{-1mm} 10 \alpha_{\mbox{\tiny LOS}} \log_{10} \| \bm{x}_u - \bm{x}_i \| + w_{u \rightarrow i}(k) ,
\\ \quad \quad  \text{if} \ m_{u \rightarrow i}=1; \\
p_{0_{\mbox{\tiny NLOS}}} \hspace{-1mm} - \hspace{-1mm} 10 \alpha_{\mbox{\tiny NLOS}} \log_{10} \| \bm{x}_u - \bm{x}_i \| + v_{u \rightarrow i}(k) , \\ \quad \quad \text{if} \ m_{u \rightarrow i}=0 ,
\end{cases}
\end{equation}
where:
\begin{itemize}
   \item $i,u$, with $u \in \Gamma_u(i)$, are indexes for the unknown nodes;
   \item $a \in \Gamma_a(i)$ is an index for anchors;
   \item $k=1,\dots,K$ is the discrete time index, with $K$ samples for each link;
   \item $p_{0_{\mbox{\tiny LOS/NLOS}}}$ is the reference power (in dBm) for the LoS or NLoS case;
   \item $\alpha_{\mbox{\tiny LOS/NLOS}}$ is the path loss exponent for the LoS or NLoS case;
   \item $\bm{x}_a$ is the known position of anchor $a$;
   \item $\bm{x}_u$ is the unknown position of agent $u$ (similarly for $\bm{x}_i$);
   \item $\Gamma_a(i)$, $\Gamma_u(i)$ are the anchor- and agent-neighborhoods of node $i$, respectively;
   \item the noise terms $w_{a \rightarrow i}(k), v_{a \rightarrow i}(k), w_{u \rightarrow i}(k), v_{u \rightarrow i}(k)$ are modeled as serially independent and identically distributed (i.i.d.), zero-mean, Gaussian random variables, independent from each other (see below), with variances:\\ $\mathrm{Var} [ w_{a \rightarrow i}(k) ] = \mathrm{Var} [ w_{u \rightarrow i}(k) ] = \sigma^2_{\mbox{\tiny LOS}}$,\\ $\mathrm{Var}[v_{a \rightarrow i}(k)] = \mathrm{Var}[v_{u \rightarrow i}(k)] = \sigma^2_{\mbox{\tiny NLOS}}$,\\ and $\sigma^2_{\mbox{\tiny NLOS}} > \sigma^2_{\mbox{\tiny LOS}} > 0$.
\end{itemize}

\noindent More precisely, letting $\delta_{i,j}$ be Kronecker's delta\footnote{$\delta_{i,j} = 1$ if and only if $i=j$, zero otherwise.}, the independence assumption is formalized by the following equations
\begin{equation}
\begin{aligned}
&\mathbb{E}[w_{j_1 \rightarrow i_1}(k_1) w_{j_2 \rightarrow i_2}(k_2) ] = \sigma^2_{\mbox{\tiny LOS}} \delta_{k_1, k_2} \delta_{i_1, i_2} \delta_{j_1, j_2} \\
&\mathbb{E}[v_{j_1 \rightarrow i_1}(k_1) v_{j_2 \rightarrow i_2}(k_2) ] \ = \sigma^2_{\mbox{\tiny NLOS}} \delta_{k_1, k_2} \delta_{i_1, i_2} \delta_{j_1, j_2} \\
&\mathbb{E}[w_{j_1 \rightarrow i_1}(k_1) v_{j_2 \rightarrow i_2}(k_2) ] = 0
\end{aligned}
\end{equation}
for any $k_1, k_2, i_1, i_2, j_1 \in \Gamma(i_1), j_2 \in \Gamma(i_2)$. The previous equations imply that two different links are always independent, regardless of the considered time instant. In this paper, we call this property \emph{link independence}. If only one link is considered, i.e. $j_2 = j_1$ and $i_2 = i_1$, then independence is preserved by choosing different time instants, implying that the sequence $\{ w_{j \rightarrow i} \}_{k} \triangleq \{ w_{j \rightarrow i}(1), w_{j \rightarrow i}(2), \dots \}$ is white. The same reasoning applies to the (similarly defined) sequence $\{ v_{j \rightarrow i} \}_{k}$. As a matter of notation, we denote the unknown positions (indexing the agents before the anchors) by $\bm{x} \triangleq [\bm{x}^\top_1 \ \cdots \ \bm{x}^\top_{N_u}]^\top \in \mathbb{R}^{2 N_u \times 1}$ and we define $\bm{\eta}$ as the collection of all channel parameters, i.e., $\bm{\eta} \triangleq [\bm{\eta}^\top_{\mbox{\tiny LOS}} \ \bm{\eta}^\top_{\mbox{\tiny NLOS}}]^\top$, with $\bm{\eta}_{\mbox{\tiny LOS}} \triangleq [p_{0_{\mbox{\tiny LOS}}} \ \alpha_{\mbox{\tiny LOS}} \ \sigma^2_{\mbox{\tiny LOS}}]^\top \in \mathbb{R}^{3 \times 1}$, $\bm{\eta}_{\mbox{\tiny NLOS}} \triangleq [p_{0_{\mbox{\tiny NLOS}}} \ \alpha_{\mbox{\tiny NLOS}} \ \sigma^2_{\mbox{\tiny NLOS}}]^\top \in \mathbb{R}^{3 \times 1}$.

It is important to stress that, in a more realistic scenario, channel parameters may vary from link to link and also across time.
However, such a generalization would produce an under-determined system of equations, thus giving up uniqueness of the solution and, more generally, analytical tractability of the problem.
For the purposes of this paper, the observation model above is sufficiently general to solve the localization task while retaining analytical tractability.

\subsection{Time-Averaged RSS Measures}
Motivated by a result given in Appendix \ref{sec:A}, we consider the time-averaged RSS measures, defined as
\begin{equation}
\bar{r}_{j \rightarrow i} \triangleq \frac{1}{K} \sum_{k=1}^K r^{(m)}_{j \rightarrow i}(k), \quad j \in \Gamma(i)
\end{equation}
as our new observables\footnote{For better readability, the notation $^{(m)}$ has not been carried over, as it implicit in the formalism.}.
While it would have been preferable to work with the original data from a theoretical standpoint, several considerations lead to the preference of time-averaged data, most notably: (1) comparison with other algorithms present in the literature, where the data model assumes only one sample per link, i.e. $K=1$, which is simply a special case in this paper; (2) reduced computational complexity in the subsequent algorithms; (3) if the RSS measures onto a given link needs to be communicated between two nodes, the communication cost is notably reduced, since only one scalar, instead of $K$ samples, needs to be communicated; (4) formal simplicity of the subsequent equations.

Moreover, from Appendix \ref{sec:A}, it follows that, assuming known $\bm{L}$, the ML estimators of the unknown positions based upon original data or time-averaged data are actually the same. To see this, it suffices to choose $\bm{\theta} = (\bm{x}, p_{0_{\mbox{\tiny LOS}}}, p_{0_{\mbox{\tiny NLOS}}}, \alpha_{\mbox{\tiny LOS}}, \alpha_{\mbox{\tiny NLOS}})$ and
\begin{equation}
s_{j \rightarrow i}(\bm{\theta}) = \begin{cases} p_{0_{\mbox{\tiny LOS}}} - 10 \alpha_{\mbox{\tiny LOS}} \log_{10} \| \bm{x}_j - \bm{x}_i \|, \\
\quad \quad \text{if} \ m_{j \rightarrow i}=1 ; \\
p_{0_{\mbox{\tiny NLOS}}} - 10 \alpha_{\mbox{\tiny NLOS}} \log_{10} \| \bm{x}_j - \bm{x}_i \|, \\
\quad \quad \text{if} \ m_{j \rightarrow i}=0
\end{cases}
\end{equation}
for $j \in \Gamma(i)$ and splitting the additive noise term as required. For a fixed link, only one of two cases (LoS or NLoS) is verified, thus, applying (\ref{eq:A}) of Appendix \ref{sec:A} yields
\begin{equation}
\label{eq:timeavg}
\begin{split}
\arg \max_{\bm{\theta}} p(r^{(m)}_{j \rightarrow i}(1) ,\dots, r^{(m)}_{j \rightarrow i}(K) ; \bm{\theta}, \sigma_{j \rightarrow i}^2) \\ = \arg \max_{\bm{\theta}} p(\bar{r}_{j \rightarrow i}; \bm{\theta}, \sigma^2_{j \rightarrow i})
\end{split}
\end{equation}
where $\sigma_{j \rightarrow i}^2$ is either $\sigma^2_{\mbox{\tiny LOS}}$ or $\sigma^2_{\mbox{\tiny NLOS}}$ and the general result follows from link independence.

We define $\mathcal{R}_i$ as the set of all RSS measures that node $i$ receives from anchor neighbors, i.e., $\mathcal{R}_i \triangleq \{ r^{(m)}_{a \rightarrow i}(1), \dots, r^{(m)}_{a \rightarrow i}(K) : a \in \Gamma_a(i) \}$, $\mathcal{Z}_i$ as the set of all RSS measures that node $i$ receives from agent neighbors, i.e., $\mathcal{Z}_i \triangleq \{ r^{(m)}_{j \rightarrow i}(1), \dots, r^{(m)}_{j \rightarrow i}(K) : j \in \Gamma_u(i) \}$, and $\mathcal{Y}_i$ as the set of all RSS measures locally available to node $i$, i.e., $\mathcal{Y}_i \triangleq \mathcal{R}_i \cup \mathcal{Z}_i$. Analogously, for time-averaged measures, we define $\bar{\mathcal{R}}_i \triangleq \{  \bar{r}_{a \rightarrow i} : a \in \Gamma_a(i) \}$, $\bar{\mathcal{Z}}_i \triangleq \{ \bar{r}_{j \rightarrow i} : j \in \Gamma_u(i) \}$, and $\bar{\mathcal{Y}}_i = \bar{\mathcal{R}}_i \cup \bar{\mathcal{Z}}_i$. Finally, we define
\begin{equation}
\Upsilon \triangleq \bigcup_{i=1}^{N_u} \mathcal{Y}_i
\end{equation}
which represents the information available to the whole network.

\subsection{Single-agent Robust Maximum Likelihood (ML)}
We first consider the single-agent case, which we will later use as a building block in the multi-agent case. The key idea is that instead of separately treating the LoS and NLoS cases, e.g. by hypothesis testing, we resort to a two-component Gaussian mixture model for the time-averaged RSS measures. More precisely, we assume that the probability density function (pdf),  $p(\cdot)$, of the time-averaged RSS measures, for anchor-agent links, is given by
\begin{multline}
\label{eq:gm1}
\hspace{-7mm} 
p(\bar{r}_{a \rightarrow i} ) = \\ 
\hspace{-7mm} 
\frac{\lambda_i}{\sqrt{2 \pi \frac{\sigma^2_{\mbox{\tiny LOS}}}{K}}} e^{ - \frac{K}{2 \sigma^2_{\mbox{\tiny LOS}}}
( \bar{r}_{a \rightarrow i} - p_{0_{\mbox{\tiny LOS}}} + 10 \alpha_{\mbox{\tiny LOS}} \log_{10} \| \bm{x}_a - \bm{x}_i \|)^2 } +  \\
\hspace{-8mm} 
\frac{1-\lambda_i}{\sqrt{2 \pi \frac{\sigma^2_{\mbox{\tiny NLOS}}}{K}}} e^{ - \frac{K}{2 \sigma^2_{\mbox{\tiny NLOS}}}
( \bar{r}_{a \rightarrow i} - p_{0_{\mbox{\tiny NLOS}}} + 10 \alpha_{\mbox{\tiny NLOS}} \log_{10} \| \bm{x}_a - \bm{x}_i \|)^2 } \\[-5mm] \hspace{-10mm}
\end{multline}
and, for agent-agent links,
\begin{multline}
\hspace{-7mm}
p(\bar{r}_{u \rightarrow i} ) = \\
\hspace{-7mm}
\frac{\zeta_i}{\sqrt{2 \pi \frac{\sigma^2_{\mbox{\tiny LOS}}}{K}}} e^{ - \frac{K}{2 \sigma^2_{\mbox{\tiny LOS}}}
( \bar{r}_{u \rightarrow i} - p_{0_{\mbox{\tiny LOS}}} + 10 \alpha_{\mbox{\tiny LOS}} \log_{10} \| \bm{x}_u - \bm{x}_i \|)^2 } +\\
\hspace{-8mm}
 \frac{1- \zeta_i}{\sqrt{2 \pi \frac{\sigma^2_{\mbox{\tiny NLOS}}}{K}}} e^{ - \frac{K}{2 \sigma^2_{\mbox{\tiny NLOS}}}
( \bar{r}_{u \rightarrow i} - p_{0_{\mbox{\tiny NLOS}}} + 10 \alpha_{\mbox{\tiny NLOS}} \log_{10} \| \bm{x}_u - \bm{x}_i \|)^2 }\\[-5mm] \hspace{-10mm}
\end{multline}
where:
\begin{itemize}
   \item $\lambda_i \in (0,1)$ is the mixing coefficient for anchor-agent links of node $i$;
   \item $\zeta_i \in (0,1)$ is the mixing coefficient for agent-agent links of node $i$.
\end{itemize}
Empirically, we can intuitively interpret $\lambda_i$ as the fraction of anchor-agent links in LoS (for node $i$), while $\zeta_i$ as the fraction of agent-agent links in LoS (for node $i$). As in~\cite{YinEM}, the Markov chain induced by our model is regular and time-homogeneous. From this, it follows that the Markov chain will converge to a two-component Gaussian mixture, giving a theoretical justification to the proposed approach.

Assume that there is a single agent, say node $i$, with a minimum of three anchors\footnote{The reason for this is that localizing a node in 2D requires at least three anchors.} in its neighborhood ($| \Gamma_a(i) | \geq 3$), in a mixed LoS/NLoS scenario. Our goal is to obtain the Maximum Likelihood estimator (MLE) of the position of node $i$. Let $\bm{\bar{r}}_i = [\bar{r}_{1 \rightarrow i} \ \cdots \ \bar{r}_{| \Gamma_a(i)| \rightarrow i}]^\top \in \mathbb{R}^{| \Gamma_a(i)| \times 1}$ be the collection of all the time-averaged RSS measures available to node $i$. Using the previous assumptions and the independency between the links, the joint likelihood function\footnote{Hereafter, we omit the conditioning on the set $\{ o_{n \rightarrow i} \}$ of actually observed RSS measures (received by node $i$) in the joint likelihood function, since it is implicit in the neighborhood formalism.} $p( \bm{\bar{r}}_i ; \bm{\theta} )$ is given by
\begin{equation}
p( \bm{\bar{r}}_i ; \bm{\theta} ) = \prod_{a \in \Gamma_a(i)} p(\bar{r}_{a \rightarrow i}; \bm{\theta} )
\end{equation}
where $\bm{\theta} = (\bm{x}_i, \lambda_i, \bm{\eta})$. Thus, denoting with $L(\bm{\theta}; \bm{\bar{r}}_i)$ the log-likelihood, we have
\begin{equation}
L(\bm{\theta}; \bm{\bar{r}}_i) = \sum_{a \in \Gamma_a(i)} \ln p(\bar{r}_{a \rightarrow i})
\end{equation}
The MLE of $\bm{\theta}$ is given by
\begin{equation}
\label{eq:robustsc}
\hat{\bm{\theta}}_{ML} = \arg \max_{\bm{\theta}} L( \bm{\theta} ;  \bm{\bar{r}}_i)
\end{equation}
where the maximization is subject to several constraints: $\lambda_i \in (0,1)$, $\alpha_{\mbox{\tiny LOS}} > 0$, $\alpha_{\mbox{\tiny NLOS}} > 0$, $\sigma^2_{\mbox{\tiny LOS}} > 0$, $\sigma^2_{\mbox{\tiny NLOS}} > 0$. In general, the previous maximization admits no closed-form solution, so we must resort to numerical procedures. \smallskip

\subsection{Multi-agent Robust ML-based scheme}
In principle, our goal would be to have a ML estimate of all the $N_u$ unknown positions, denoted by $\bm{x}$. Let $\bm{\lambda} \triangleq [\lambda_1 \ \cdots \ \lambda_{N_u}]^\top$, $\bm{\zeta} \triangleq [\zeta_1 \ \cdots \ \zeta_{N_u}]^\top$ be the collections of the mixing coefficients. Defining $\bm{\theta} = (\bm{x}, \bm{\lambda}, \bm{\zeta}, \bm{\eta})$, the ML joint estimator
\begin{equation}
\hat{\bm{\theta}}_{ML} = \arg \max_{\bm{\theta}} p( \Upsilon; \bm{\theta})
\end{equation}
is, in general, computationally unfeasible and naturally centralized. In order to obtain a practical algorithm, we now resort to a sub-optimal, but computationally feasible and distributed approach. The intuition is as follows. Assume, for a moment, that a specific node $i$ knows $\bm{\eta}$, $\bm{\lambda}$, $\bm{\zeta}$ and also all the true positions of its neighbors (which we denote by $\mathcal{X}_i$). Then, the ML joint estimation problem is notably reduced, in fact,
\begin{equation}
\hat{\bm{x}}_{i_{ML}} = \arg \max_{\bm{x}_i} p( \Upsilon; \bm{x}_i)
\end{equation}
We now make the sub-optimal approximation of avoiding non-local information in order to obtain a distributed algorithm, thus resorting to
\begin{equation}
\label{eq:multiapprox}
\hat{\bm{x}}_i = \arg \max_{\bm{x}_i} p( \mathcal{\bar{Y}}_i ; \bm{x}_i, \mathcal{X}_i, \lambda_i, \zeta_i, \bm{\eta} )
\end{equation}
where we made explicit the functional dependence on all the other parameters (which, for now, are assumed known). Due to the i.i.d.~hypothesis, the ``local'' likelihood function has the form
\begin{multline}
p( \mathcal{\bar{Y}}_i ; \bm{x}_i, \mathcal{X}_i, \lambda_i, \zeta_i, \bm{\eta} ) = \prod_{a \in \Gamma_a(i)} p( \bar{r}_{a \rightarrow i} ; \bm{x}_i, \lambda_i, \bm{\eta} )\times\\ \prod_{j \in \Gamma_u(i) } p ( \bar{r}_{j \rightarrow i} ; \bm{x}_i, \mathcal{X}_i, \zeta_i, \bm{\eta} )
\end{multline}
where the marginal likelihoods are Gaussian-mixtures and we underline the (formal and conceptual) separation between anchor-agent links and agent-agent links. By taking the natural logarithm, we have
\begin{multline}
\hspace{-7mm}
\ln p( \mathcal{\bar{Y}}_i ;  \bm{x}_i, \mathcal{X}_i, \lambda_i, \zeta_i, \bm{\eta} ) = \sum_{a \in \Gamma_a(i)} \ln p( \bar{r}_{a \rightarrow i} ; \bm{x}_i, \lambda_i, \bm{\eta} ) + \\
\sum_{j \in \Gamma_u(i) } \ln p ( \bar{r}_{j \rightarrow i} ; \bm{x}_i, \mathcal{X}_i, \zeta_i, \bm{\eta} )
\end{multline}
The maximization problem in~(\ref{eq:multiapprox}) then reads
\begin{multline}
\label{eq:theocost}
\hat{\bm{x}}_i = \arg \max_{\bm{x}_i} \left\{ \sum_{a \in \Gamma_a(i)} \ln p( \bar{r}_{a \rightarrow i} ; \bm{x}_i, \lambda_i, \bm{\eta} ) \right.+ \\
\left.\sum_{j \in \Gamma_u(i) } \ln p ( \bar{r}_{j \rightarrow i} ; \bm{x}_i, \mathcal{X}_i, \zeta_i, \bm{\eta} ) \right\}.
\end{multline}
We can now relax the initial assumptions: instead of assuming known neighbors positions $\mathcal{X}_i$, we will substitute them with their estimates, $\hat{\mathcal{X}}_i$. Moreover, since the robust ML-based self-calibration can be done without knowing the channel parameters $\bm{\eta}$, we also maximize over them. Lastly, we maximize with respect to the mixing coefficients $\lambda_i, \zeta_i$. Thus, our final approach is
\begin{multline}
\label{eq:robustcost}
\hspace{-8mm}
(\hat{\bm{x}}_i, \hat{\lambda}_i, \hat{\zeta}_i, \hat{\bm{\eta}}) = \arg \max_{\bm{x}_i, \lambda_i, \zeta_i, \bm{\eta}} \left\{ \sum_{a \in \Gamma_a(i)} \ln p( \bar{r}_{a \rightarrow i} ; \bm{x}_i, \lambda_i, \bm{\eta} ) \right. \\
+ \left. \sum_{j \in \hat{\Gamma}_u(i) } \ln p ( \bar{r}_{j \rightarrow i} ; \bm{x}_i, \hat{\mathcal{X}}_i, \zeta_i, \bm{\eta} ) \right\}
\end{multline}
where $\hat{\Gamma}_u(i)$ is the set of all agent-neighbors of node $i$ for which estimated positions exist. We can iteratively construct (and update) the set $\hat{\Gamma}_u(i)$, in order to obtain a fully distributed algorithm, as summarized in Algorithm 1.
\begin{algorithm}[h]
 \caption{Robust Distributed Maximum Likelihood (RD-ML)}
 \begin{algorithmic}[0]
 \STATE \textbf{Initialization}:
 \STATE \emph{Self-Localization:} Every node $i$ with $|\Gamma_a(i)| \geq 3$ self-localizes by the robust procedure~(\ref{eq:robustsc});
 \STATE \emph{Broadcast:} Every node which self-localized broadcasts its position estimate $\hat{\bm{x}}_i$;
\\ \textbf{Iterative scheme}: Start with $n \leftarrow 1$;
  \STATE \emph{Update position}: Every node with $| \Gamma_a(i)| + |\hat{\Gamma}^{(n)}_u(i) | \geq 3$ estimates its own position by solving~(\ref{eq:robustcost});
  \STATE \emph{Broadcast:} The updated position $\hat{\bm{x}}^{(n)}_i$ is broadcasted;
  \STATE \emph{Update Neighbors:} Every node $j$ for which the new estimated positions $\hat{\bm{x}}^{(n)}_i$ are neighbors updates its own $\hat{\Gamma}^{(n)}_u(j)$ with the new positions;
   \STATE \emph{Repeat:} Set $n \leftarrow n+1$ and repeat the previous steps until stop condition is met.
 \\ \textbf{Stop condition:}
 \STATE $| \hat{\Gamma}^{(n)}_u(i) | = | \Gamma_u(i) |$ for every node $i$.
 \end{algorithmic}
 \end{algorithm}
 
A few remarks are now in order. First, this algorithm imposes some restrictions on the arbitrariness of the network topology, since the information spreads starting from the agents which were able to self-localize during initialization; in practice, this requires the network to be sufficiently connected. Second, convergence of the algorithm is actually a matter of compatibility: if the network is sufficiently connected (compatible), convergence is guaranteed. Given a directed graph, compatibility can be tested a priori and necessary and sufficient conditions can be found (see Section IV). Third, unlike many algorithms present in the literature, symmetrical links are not necessary, nor do we resort to symmetrization (like NBP): this algorithm naturally takes into account the (possible) asymmetrical links of directed graphs.

\subsection{Distributed Maximum Likelihood (D-ML)}
As a natural competitor of the proposed RD-ML algorithm, we derive here the Distributed Maximum Likelihood (D-ML) algorithm, which assumes that all links are of the same type. As its name suggests, this is the non-robust version of the previously derived RD-ML. As usual, we start with the single-agent case as a building block for the multi-agent case. Using the assumption that all links are the same and the i.i.d. hypothesis, the joint pdf of the time-averaged RSS measures, received by agent $i$, is given by
\begin{multline}
p(\bm{\bar{r}}_i ; \bm{x}_i, p_0, \alpha, \sigma^2) = \frac{1}{(2 \pi \sigma^2)^{|\Gamma_a(i)|/2}} \times \\
e^{ -\frac{1}{2 \sigma^2} \sum_{a \in \Gamma_a(i)} ( \bar{r}_{a \rightarrow i} - p_0 + 10 \alpha \log_{10} \| \bm{x_i} - \bm{x}_a \| )^2 }
\end{multline}
We can now proceed by estimating, with the ML criterion, first $p_0$ as a function of the remaining parameters, followed by $\alpha$ as a function of $\bm{x}_i$ and finally $\bm{x}_i$. We have
\begin{multline}
\hspace{-7mm}
\hat{p}_0(\alpha, \bm{x}_i) = \\ 
\hspace{-8mm} \arg \min_{p_0} \sum_{a \in \Gamma_a(i)} ( \bar{r}_{a \rightarrow i} - p_0 + 10 \alpha \log_{10} \| \bm{x_i} - \bm{x}_a \| )^2. \\[-9mm] \hspace{-7mm}
\end{multline}\\
Defining $s_{a,i} \triangleq 10 \log_{10} \| \bm{x_i} - \bm{x}_a \|$ as the log-distance, $\bm{s}_i \triangleq [s_{1,i} \ s_{2,i} \ \cdots \ s_{|\Gamma_a(i)|,i}]^\top \in \mathbb{R}^{|\Gamma_a(i)| \times 1}$ the column-vector collecting them and $\bm{1}_n = [1 \cdots 1]^\top \in \mathbb{R}^{n \times 1}$ an all-ones vector of dimension $n$, the previous equation can be written as
\begin{equation}
\hat{p}_0(\alpha, \bm{x}_i) = \arg \min_{p_0} \| \bm{\bar{r}}_i  + \alpha \bm{s}_i - p_0 \bm{1}_{|\Gamma_a(i)|} \|^2
\end{equation}
which is a Least-Squares (LS) problem and its solution is
\begin{equation}
\label{eq:dmlp0}
\hat{p}_0(\alpha, \bm{x}_i) = \frac{1}{| \Gamma_a(i) |} \sum_{a \in \Gamma_a(i)} (\bar{r}_{a \rightarrow i}  + 10 \alpha \log_{10} \| \bm{x_i} - \bm{x}_a \| )
\end{equation}
By using this expression, the problem of estimating $\alpha$ as a function of $\bm{x}_i$ is
\begin{equation}
\label{eq:alphamin}
\hat{\alpha}(\bm{x}_i) = \arg \min_{\alpha} \| \bm{P^{ ^\perp}}_{\bm{1}_{|\Gamma_a(i)|}}( \bm{\bar{r}}_i + \alpha \bm{s}_i) \|^2
\end{equation}
where, given a full-rank matrix $\bm{A} \in \mathbb{R}^{m \times n}$, with $m \geq n$, $\bm{P}^{^ \perp}_{\bm{A}}$ is the orthogonal projection matrix onto the orthogonal complement of the space spanned by the columns of $\bm{A}$. It can be computed via $\bm{P}^{^ \perp}_{\bm{A}} = \bm{I}_m - \bm{P_A}$, where $\bm{P_A} = \bm{A} (\bm{A}^\top \bm{A})^{-1} \bm{A}^\top$ is an orthogonal projection matrix and $\bm{I}_m$ is the identity matrix of order $m$. The solution to problem~(\ref{eq:alphamin}) is given by
\begin{equation}
\label{eq:dmlalpha}
\hat{\alpha}(\bm{x}_i) = - \bm{ ( \tilde{s}_i^\top \tilde{s}_i )^{-1} \tilde{s}_i^\top \tilde{r}_i } 
\end{equation}
where $\bm{\tilde{r}}_i = \bm{P^{ ^\perp}}_{\bm{1}_{|\Gamma_a(i)|}} \bm{\bar{r}}_i$ and $\bm{\tilde{s}}_i = \bm{P^{ ^\perp}}_{\bm{1}_{|\Gamma_a(i)|}} \bm{s}_i$. By using the previous expression, we can finally write
\begin{equation}
\label{eq:dmlsingle}
\hat{\bm{x}_i} = \arg \min_{\bm{x}_i} \| \bm{P^{ ^\perp}}_{\bm{\tilde{s}_i}} \bm{\tilde{r}}_i \|^2
\end{equation}
which, in general, does not admit a closed-form solution, but can be solved numerically. After obtaining $\hat{\bm{x}}_i$, node $i$ can estimate $p_0$ and $\alpha$ using (\ref{eq:dmlp0}) and (\ref{eq:dmlalpha}).

The multi-agent case follows an almost identical reasoning of the RD-ML. Approximating the true (centralized) MLE by avoiding non-local information and assuming to already have an initial estimate of $p_0$ and $\alpha$, it is possible to arrive at
\begin{multline}
\label{eq:dmlcost}
\hspace{-7mm}
\hat{\bm{x}}_i = \arg \min_{\bm{x}_i} \Bigg\{ \sum_{a \in \Gamma_a(i)}( \bar{r}_{a \rightarrow i} - p_0 + 10 \alpha \log_{10} \| \bm{x_i} - \bm{x}_a \| )^2 \\
\hspace{-7mm}
+ \sum_{j \in \hat{\Gamma}_u(i)} ( \bar{r}_{j \rightarrow i} - p_0 + 10 \alpha \log_{10} \| \bm{x_i} - \bm{x}_j \| )^2 \Bigg\}
\end{multline}
where (again) an initialization phase is required and the set of estimated agents-neighbors $\hat{\Gamma}_u(i)$ is iteratively updated. The key difference with RD-ML is that, due to the assumption of the links being all of the same type, the estimates of $p_0$ and $\alpha$ are broadcasted and a common consensus is reached by averaging. This increases the communication overhead, but lowers the computational complexity, operating a trade-off. The D-ML algorithm is summarized in Algorithm 2.
\begin{algorithm}[h]
 \caption{Distributed Maximum Likelihood (D-ML)}
 \begin{algorithmic}[0]
 \STATE \textbf{Initialization}:
 \STATE \emph{Self-Localization:} Every node $i$ with $|\Gamma_a(i)| \geq 3$ self-localizes by solving~(\ref{eq:dmlsingle}) and estimates $\alpha$ via~(\ref{eq:dmlalpha}) and $p_0$ via~(\ref{eq:dmlp0}), thus obtaining $(\hat{p}^{(i)}_0, \hat{\alpha}^{(i)}, \hat{\bm{x}}_i)$;
 \STATE \emph{Broadcast:} Every node which self-localized broadcasts its local estimates $(\hat{p}^{(i)}_0, \hat{\alpha}^{(i)}, \hat{\bm{x}}_i)$;
 \STATE \emph{Consensus:} All the nodes agree on global values of $(\hat{p}_0, \hat{\alpha})$ by averaging all the local estimates $(\hat{p}^{(i)}_0, \hat{\alpha}^{(i)})$;
\\ \textbf{Iterative scheme}: Start with $n \leftarrow 1$;
  \STATE \emph{Update position}: Every node $i$ with $| \Gamma_a(i)| + |\hat{\Gamma}^{(n)}_u(i) | \geq 3$ estimates its own position by solving~(\ref{eq:dmlcost});
  \STATE \emph{Broadcast:} The updated position $\hat{\bm{x}}^{(n)}_i$ is broadcasted;
  \STATE \emph{Update Neighbors:} Every node $j$ for which the new estimated positions $\hat{\bm{x}}^{(n)}_i$ are neighbors updates its own $\hat{\Gamma}^{(n)}_u(j)$ with the new positions;
   \STATE \emph{Repeat:} Set $n \leftarrow n+1$ and repeat the previous steps until stop condition is met.
 \\ \textbf{Stop condition:}
 \STATE $| \hat{\Gamma}^{(n)}_u(i) | = | \Gamma_u(i) |$ for every node $i$.
 \end{algorithmic}
 \end{algorithm}
 
Similar remarks as for the RD-ML can be made for the D-ML. Again, the network's topology cannot be completely arbitrary, as the information must spread throughout the network starting from the agents which self-localized, implying that the graph must be sufficiently connected. Necessary and sufficient conditions to answer the compatibility question are the same as RD-ML. Secondly, the (strong) hypothesis behind the D-ML derivation (i.e., all links of the same type) allows for a more analytical derivation, up to position estimation, which is a Nonlinear Least-Squares problem. However, it is also its weakness since, as will be shown later, it is not a good choice for mixed LoS/NLoS scenarios.

\subsection{Centralized MLE with known nuisance parameters (C-MLE)}
The centralized MLE of $\bm{x}$ with known nuisance parameters, i.e. assuming known $\bm{L}$ and $\bm{\eta}$, is chosen here as a benchmark for both RD-ML and D-ML. In the following, this algorithm will be denoted by C-MLE. Its derivation is simple (see Appendix \ref{sec:cmle}) and results in
\begin{multline}
\label{eq:cmle}
\hat{\bm{x}}_{ML} = \arg \min_{\bm{x}} \sum_{i=1}^{N_u} \sum_{j \in \Gamma(i)} \frac{1}{\sigma^2_{j \rightarrow i}} ( \bar{r}_{j \rightarrow i}
- p_{0_{j\rightarrow i}} + \\
10 \alpha_{j \rightarrow i} \log_{10} \| \bm{x}_j - \bm{x}_i \| )^2
\end{multline}
where $(p_{0_{j \rightarrow i}}, \alpha_{j \rightarrow i}, \sigma^2_{j \rightarrow i})$ are either LoS or NLoS depending on the considered link. It is important to observe that, if all links are of the same type, the dependence from $\sigma^2_{j \rightarrow i}$ in~(\ref{eq:cmle}) disappears. From standard ML theory~\cite{Kay1}, C-MLE is asymptotically ($K\rightarrow +\infty$) optimal. The optimization problem~(\ref{eq:cmle}) is computationally challenging, as it requires a minimization in a $2 N_u$-dimensional space, but still feasible for small values of $N_u$.

\section{Convergence Analysis}
The convergence test of our proposed algorithm (and also of D-ML) can be restated as a graph coloring problem: if all the graph can be colored, then it is compatible and convergence is guaranteed. As it is common in the literature on graph theory, let $G = (V, E)$ be a directed graph, with $V$ denoting the set of \emph{nodes} and $E$ the set of \emph{directed edges}. The set of nodes is such that $V = V_a \cup V_u$, where $V_a$ is the (nonempty) set of anchor nodes and $V_u$ is the (nonempty) set of agent nodes.

\begin{defn}[RDML-initializable]{A directed graph $G$ is said to be \emph{RDML-initializable} if and only if there exists at least one agent node, say $x$, such that $| \Gamma_a(x) | \geq 3$.}
\end{defn}

As can be easily checked, the previous statement is a necessary condition: if a graph is not RDML-initializable, then it is incompatible with RD-ML. To give a necessary and sufficient condition, we introduce the notion of ``color''. A node can be either black or white; all anchors are black and all agent nodes start as white, but may become black if some condition is satisfied. The RD-ML can be rewritten as a graph coloring problem. In order to do this, we define the set $\hat{\Gamma}_u(i) \triangleq \{ j \in \Gamma_u(i) : \text{agent} \ j \ \text{is black} \}$, i.e., the subset of agent-neighbors of node $i$ which contains only black agent nodes. In general, $\hat{\Gamma}_u(i) \subseteq \Gamma_u(i)$. We also define the set $B_u \triangleq \{ i \in V_u: \text{agent $i$ is black} \}$, i.e., the set of black agents. In general, $B_u \subseteq V_u$. Given a graph $G$, we can perform a preliminary test by running the following RDML-coloring algorithm (a better test will be derived later):
\begin{itemize}
   \item \textbf{Initialization} ($k=0$)
   \begin{enumerate}
      \item All anchors are colored black and all agents white;
      \item Every agent $i$ with $|\Gamma_a(i)| \geq 3$ is colored black;
   \end{enumerate}
   \item \textbf{Iterative coloring}: Start with $k=1$
   \begin{enumerate}
      \item Every agent with $| \Gamma_a(i)| + |\hat{\Gamma}^{(k-1)}_u(i) | \geq 3$, where $\hat{\Gamma}^{(k)}_u(i)$ is the set $\hat{\Gamma}_u(i)$ at step $k$, is colored black;
      \item Every agent $j$ updates is own $\hat{\Gamma}^{(k)}_u(j)$ with the new colored nodes;
      \item The set $B^{(k)}_u$ is updated, where $B^{(k)}_u$ is the set $B_u$ at step $k$;
      \item Set $k \leftarrow k+1$ and repeat the previous steps until $V_u$ contains only black nodes.
   \end{enumerate}
\end{itemize}
Suppose that the previous algorithm can color the entire graph black in a finite amount of steps, say $n$. Then, $n$ is called RDML-lifetime.

\begin{defn}[RDML-lifetime]{A directed graph $G$ is said to have \emph{RDML-lifetime} equal to $n$ if and only if the RDML-coloring algorithm colors black the set $V_u$ in exactly $n$ steps. If no such integer exists, by convention, $n = +\infty$.}
\end{defn}

This allows us to formally define compatibility:

\begin{defn}[RDML-compatibility]{A directed graph $G$ is said to be \emph{RDML-compatible} if and only if:
\begin{enumerate}
   \item $G$ is RDML-initializable;
   \item the RDML-lifetime of $G$ is finite.
\end{enumerate}
Otherwise, $G$ is said to be \emph{RDML-incompatible}.}
\end{defn}

In practice, there are only two ways for which a graph is RDML-incompatible: either $G$ cannot be initialized, or the RDML-lifetime of $G$ is infinite. Testing the first condition is trivial; the interesting result is that testing the second condition is made simple thanks to the following

\begin{thm}{An RDML-initializable graph $G$ is RDML-incompatible if and only if there exist an integer $h$ such that
\begin{equation}
\begin{gathered}
\label{eq:b}
B^{(h)}_u = B^{(h-1)}_u \ \wedge \ \ |B^{(h)}_u | < | V_u |
\end{gathered}
\end{equation}
that is, if there is a step $h$ in which no more agents can be colored black and at least one agent is still left white.}
\end{thm}
\begin{proof} $(\Leftarrow)$ First, observe that, by construction, \\$B^{(k-1)}_u \subseteq B^{(k)}_u$, as black nodes can only be added. Let $C_G : \mathbb{N} \rightarrow \mathbb{N}$ be the following function
\begin{equation}
C_G(k) = | B^{(k)}_u |.
\end{equation}
Since $B^{(k-1)}_u \subseteq B^{(k)}_u$, $C_G$ is non-decreasing. An RDML-initializable graph is RDML-compatible if and only if it has finite RDML-lifetime, i.e., there must exist $n$ such that $C_G(n) = | V_u |$. But condition~(\ref{eq:b}) implies that there exists $h$ such that $C_G(h) = C_G(h-1) < | V_u |$. At step $h+1$ and all successive steps, $C_G$ cannot increase since the set $B^{(k)}_u$ cannot change. To show this, the key observation is that, as the graph $G$ at step $h-1$ did not satisfy the conditions for $B^{(h-1)}_u$ to grow (by hypothesis), the set was equal to itself at step $h$, i.e., $B^{(h-1)}_u = B^{(h)}_u$. But since no color change happened in $V_u$ at step $h$, the graph $G$ still does not satisfy the conditions for $B^{(k)}_u$ to grow for $k \geq h$. Thus, $C_G(k)$ becomes a constant function for $k \geq h$ and can never reach the value $|V_u|$.

$(\Rightarrow)$ Since $G$ is an RDML-incompatible graph by hypothesis, at least one agent must be white, so $|B_u^{(h)}| < |V_u|$ is true for any $h$. Since $C_G(k)$ is non-decreasing, it must become constant for some $h > k$, but this implies that, for some $h$, $|B^{(h)}_u| = |B^{(h-1)}_u|$. This implies that, since $B^{(k-1)}_u \subseteq B^{(k)}_u$ by construction, $B^{(h)}_u = B^{(h-1)}_u$ for some $h$.
\end{proof}

\begin{defn}[RDML-depth]{Let $G$ be a directed graph. Then,
\begin{equation}
h_G \triangleq \inf_{h} \{ h \in \mathbb{N} : B^{(h)}_u = B^{(h+1)}_u \}
\end{equation}
is called the \emph{RDML-depth} of $G$.}
\end{defn}

A complete graph has $h_G=0$, as all agents are colored black during the initialization phase of the RDML-coloring algorithm.

\begin{cor}{Let $G$ be a directed graph. Then, $h_G \leq n$, where $n$ is the RDML-lifetime of $G$.}
\end{cor}
\begin{proof} If $G$ is not RDML-initializable, $h_G = 0$ as $B^{(0)}_u = \emptyset$. If $G$ is RDML-initializable, there are two cases: either $n$ is finite or not. In the latter, $n=+\infty$ and $h_G$ is finite by previous theorem. If $n$ is finite, $h_G = n$ since $|B^{(n)}_u| = |V_u|$ by definition of RDML-lifetime.
\end{proof}

The previous corollary proves that $h_G$ is always finite, regardless of $G$. This allows us to write the graph compatibility test, shown in Algorithm 3. Thanks to the previous results, this algorithm always converges and can be used to test a priori if a graph is RDML-compatible or not.

\noindent {\bf Remark} Algorithm 3 can be intuitively explained via a physical metaphor, where, in a metal grid (representing the graph), ``heat'' (information) spreads out starting from some initial ``hot spots'' (nodes that are colored black in the first iteration). This spreading is continued, reaching more and more locations on the grid, until the event occurs that further spreading of ``heat'' does not change the ``heat map''. If, at this point, there are cold spots (nodes that have not been colored black), the graph is RDML-incompatible. By contrast, if heat spreads throughout the grid, the graph is RDML-initializable.

\begin{algorithm}[h]
 \caption{Graph Compatibility Test}
 \begin{algorithmic}[1]
 \renewcommand{\algorithmicrequire}{\textbf{Input:}}
 \renewcommand{\algorithmicensure}{\textbf{Output:}}
 \REQUIRE Directed Graph $G$ with $V = V_a \cup V_u$
 \ENSURE  Binary flag $f_G$, Depth $h_G$
 \\ \textit{Initialization} :
  \STATE Set $Color(V_u) \leftarrow 0$; \%\textit{white}
  \FOR {$i=1$ to $|V_u|$}
   \IF {$|\Gamma_a(i) | \geq 3$}
    \STATE $Color(i) \leftarrow 1$; \%\textit{black}
  \ENDIF
  \ENDFOR
  \STATE Initialize set $\hat{\Gamma}^{(0)}_u(i)$ for every agent $i$ and set $B^{(0)}_u$;
  \STATE Compute $C_G(0) = |B^{(0)}_u|$; \textbf{if} $C_G(0)=|V_u|$, \textbf{stop} and \textbf{return} $(f_G, h_G)=(1,0)$;
 \\ \textit{Coloring Loop}
 \STATE Set $k \leftarrow 1$;
  \WHILE {$(1)$}
  \FOR {$i=1$ to $|V_u|$}
  \IF {$| \Gamma_a(i)| + |\hat{\Gamma}^{(k-1)}_u(i) | \geq 3$}
  \STATE $Color(i) \leftarrow 1$;
  \ENDIF
  \ENDFOR
  \STATE Update $\hat{\Gamma}^{(k)}_u(i)$ for every agent $i$ and compute $C_G(k) \leftarrow |B^{(k)}_u|$;
  \\ \textit{Exit condition}
  \IF {$C_G(k-1) = C_G(k)$}
   \STATE $h_G \leftarrow k-1$;
   \IF{$C_G(k) < |V_u|$}
    \STATE $f_G \leftarrow 0$; \%\textit{graph $G$ is RDML-incompatible}
   \ELSE
    \STATE $f_G \leftarrow 1$; \%\textit{graph $G$ is RDML-compatible}
   \ENDIF
  \STATE \textbf{break} loop;
  \ENDIF
  \STATE $k \leftarrow k+1$;
  \ENDWHILE
 \RETURN $f_G, h_G$.
 \end{algorithmic} 
 \label{algo:compat_test}
 \end{algorithm}

\begin{figure}[t!]
\centering
\includegraphics[scale=0.6]{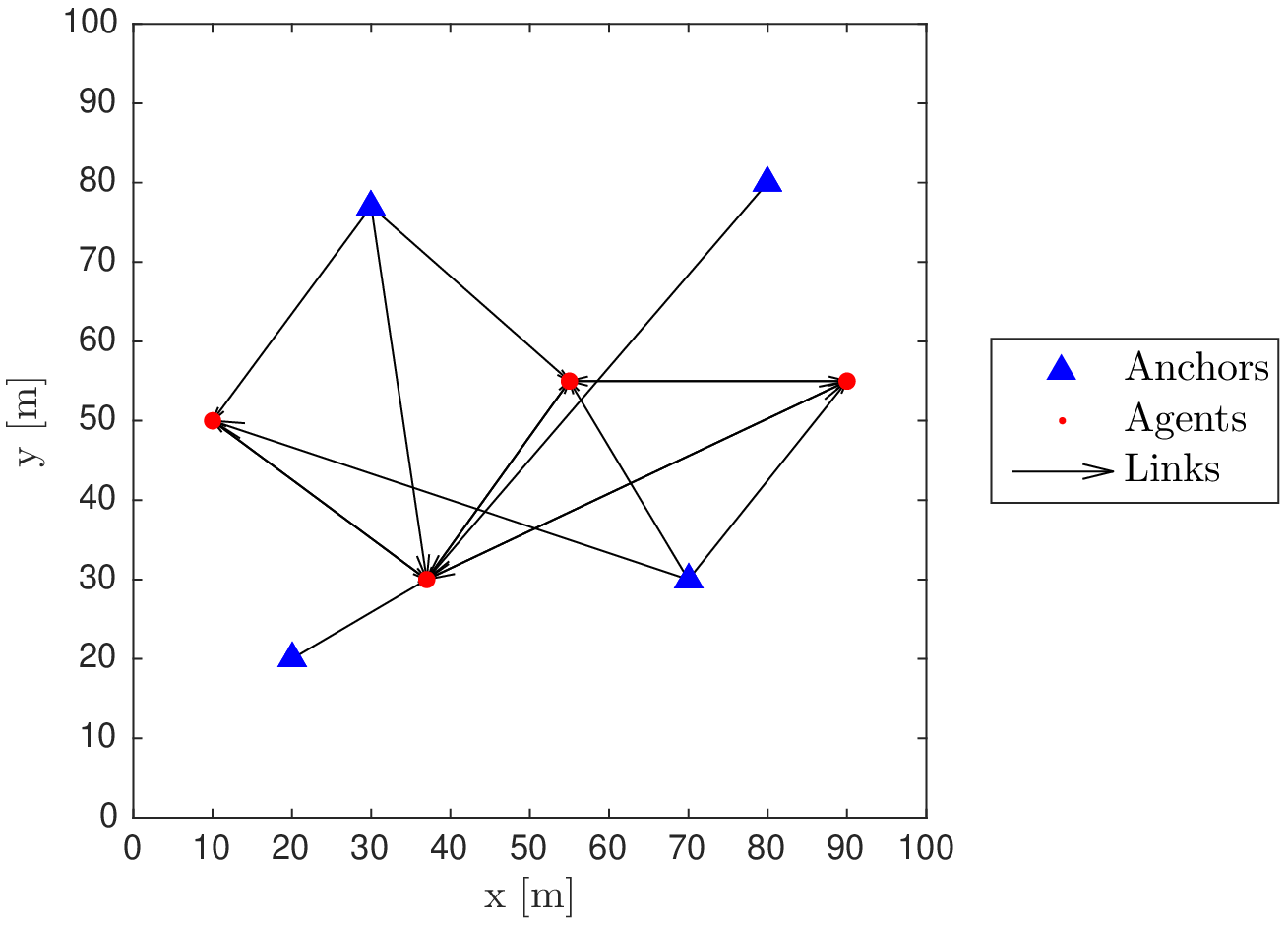}
\caption{Example of a compatible graph with $N_a=4, N_u=4$.}
\label{fig:example1}
\end{figure}
\begin{figure}[t!]
\centering
\includegraphics[scale=0.6]{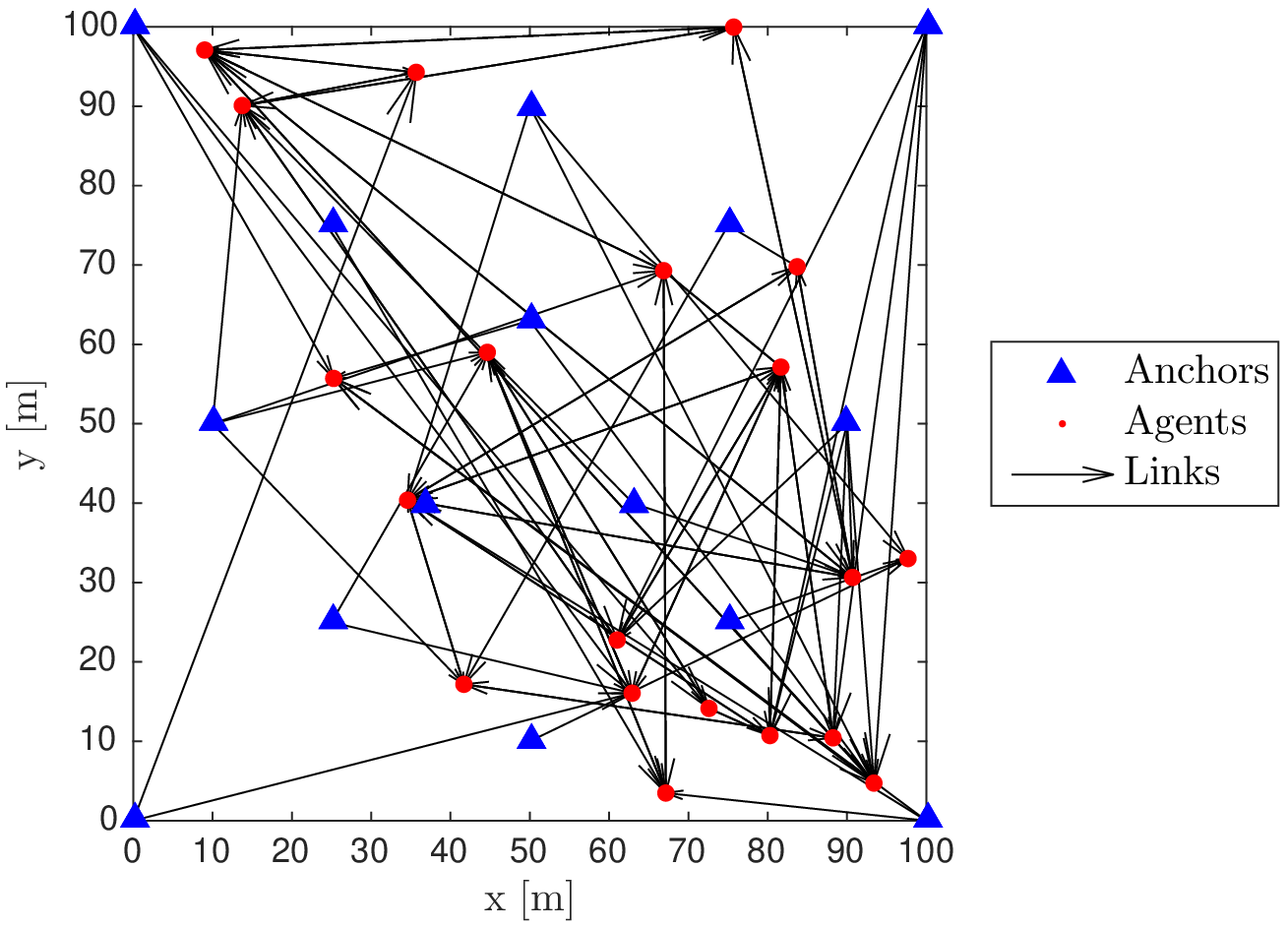}
\caption{Example of a compatible graph with $N_a=15, N_u=20$.}
\label{fig:example2}
\end{figure}
Figs.~\ref{fig:example1}-\ref{fig:example2} show two examples of RDML-initializable graphs. Fig.~\ref{fig:example1} illustrates the case of a small ($N_a=4, N_u=4$) but highly connected network. In contrast, Fig.~\ref{fig:example2} displays a larger ($N_a=15, N_u=20$), but weakly connected network (Graph Depth = 7). In particular, only $90$ out of $680$ possible directed links are connected, where anchors have $3$-$4$ links (out of $20$). For both cases, graph compatibility can be shown using Algorithm~\ref{algo:compat_test}.

\section{Results and Discussion}
In this section, we use standard Monte Carlo (MC) simulation to evaluate the performance of the proposed algorithm and its competitors. As a performance metric, we will show the ECDF (Empirical Cumulative Distribution Function) of the localization error, defined as $e_i \triangleq \| \hat{\bm{x}}_i - \bm{x}_i \|$ for agent $i$, i.e. an estimate of the probability $\mathcal{P}\{ \| \hat{\bm{x}}_i - \bm{x}_i \| \leq \gamma \}$. The ECDFs are obtained by stacking all the localization errors for every agent in a single vector, in order to give a global picture of the algorithm performances. The simulated scenario is as follows. In a square coverage area of $100 \times 100 \ \mathrm{m}^2$, $N_a = 11$ stationary anchors are deployed, as depicted in Fig.~\ref{fig:ancore}. The channel parameters are generated as follows: $p_{0_{\mbox{\tiny LOS}}} \sim \mathcal{U}[-30, 0]$, $\alpha_{\mbox{\tiny LOS}} \sim \mathcal{U}[2,4]$, $\sigma_{\mbox{\tiny LOS}} = 6$, while, for the NLoS case, $p_{0_{\mbox{\tiny NLOS}}} \sim \mathcal{N}(0, 25)$, $\alpha_{\mbox{\tiny NLOS}} \sim \mathcal{U}[3,6]$, $\sigma_{\mbox{\tiny NLOS}} = 12$. Similar settings on the reference power and the path loss exponent can be found in \cite{Mihaylova2011} and \cite{Beck2017}, respectively. At each MC trial, $N_u = 10$ agents are randomly (uniformly) generated in the coverage area. Unless stated otherwise, $K=40$ samples per link are used. Finally, each simulation consists of 100 MC trials.

The optimization problems~(\ref{eq:dmlcost}) and~(\ref{eq:robustcost}) have been solved as follows. For D-ML, a 2D \texttt{grid search} has been used, while, for RD-ML and C-MLE, the optimization has been performed with the \textsc{MATLAB} solver \texttt{fmincon}.

\begin{figure}[tp]
\centering
\includegraphics[scale=0.6]{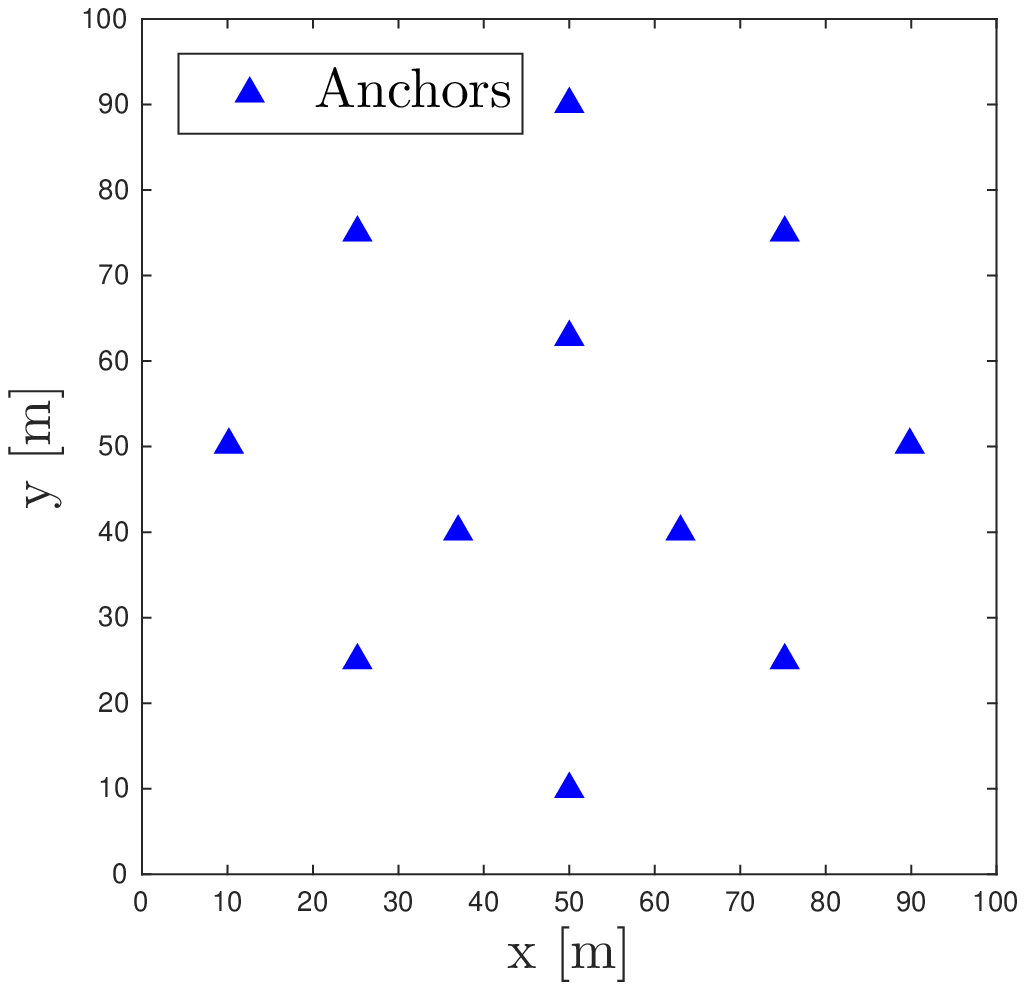}
\caption{Anchors positions used in all simulations.}
\label{fig:ancore}
\end{figure}

\subsection{Gaussian Noise}
Here, we validate our novel approach by considering fully-connected networks in three different scenarios. In Fig.~\ref{fig:los_iid}, all links are assumed to be LoS. As can be seen from the ECDFs, the robust approach has similar performances to the D-ML algorithm, which was designed assuming all links to be of the same type. In Fig.~\ref{fig:nlos_iid}, all links are assumed to be NLoS and again we see a similar result. The key difference is in Fig.~\ref{fig:mixed_iid}, where a mixed LoS/NLoS scenario is considered and the fraction of NLoS links is random (numerically generated as $\mathcal{U}[0,1]$ and different for each MC trial). Here, the robust approach outperforms D-ML, validating our strategy. Moreover, it has a remarkably close performance to the centralized algorithm C-MLE. As a result, the robust approach is the preferred option in all scenarios.

\begin{figure}[tp]
\centering
\includegraphics[scale=0.6]{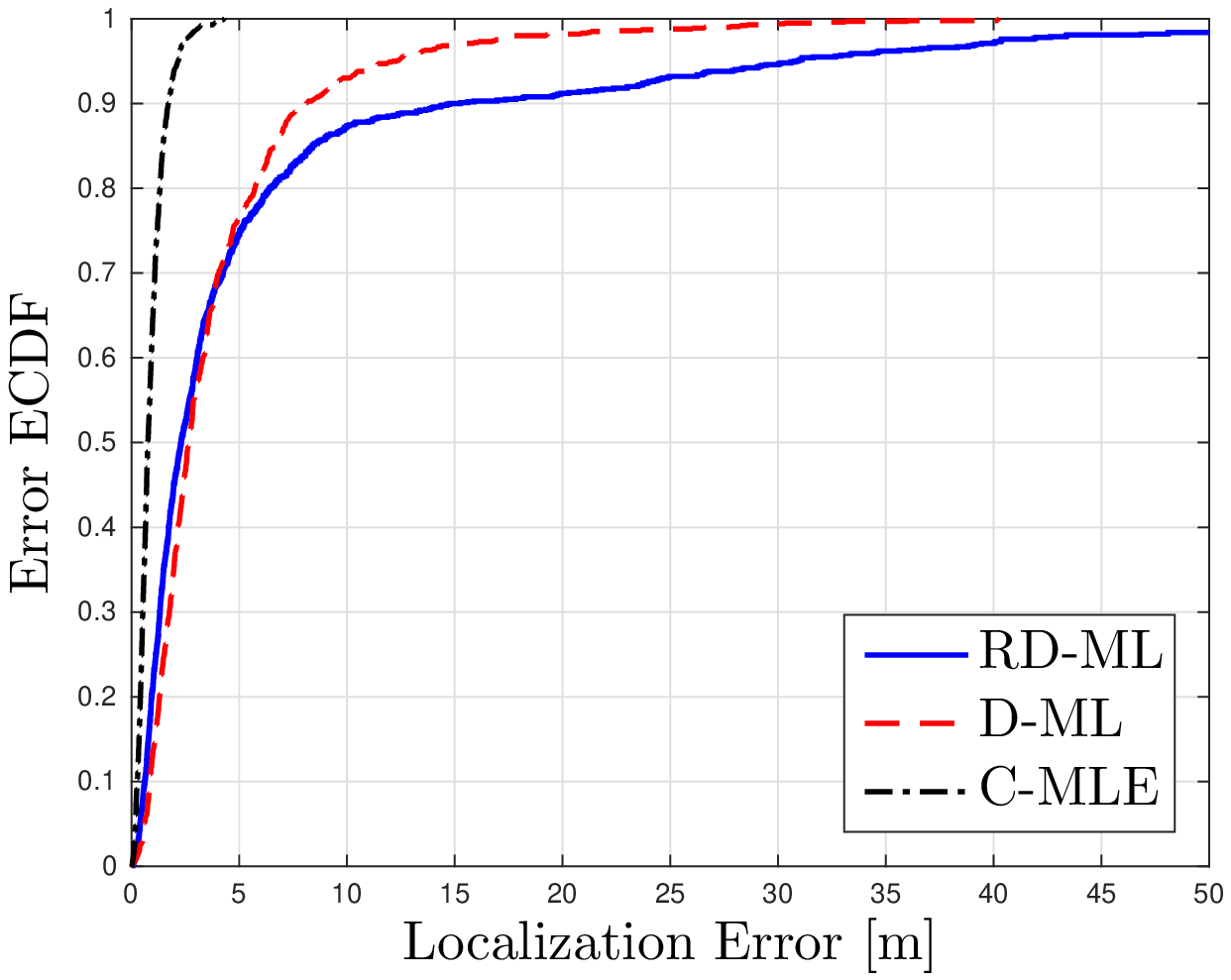}
\caption{All LoS links with white Gaussian noise. Fully-connected networks, $N_a=11, N_u=10$.}
\label{fig:los_iid}
\end{figure}

\begin{figure}[tp]
\centering
\includegraphics[scale=0.6]{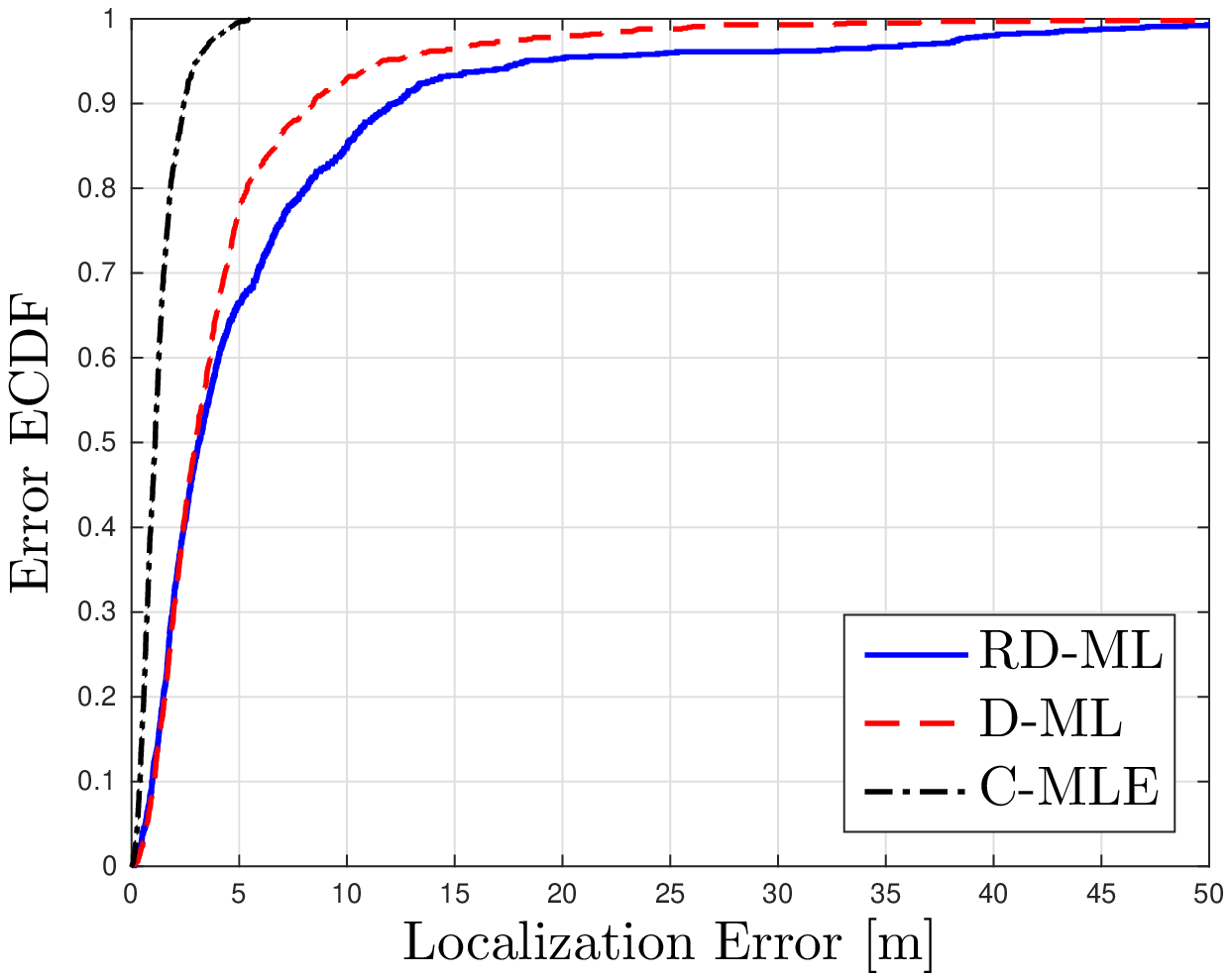}
\caption{All NLoS links with white Gaussian noise. Fully-connected networks, $N_a=11, N_u=10$.}
\label{fig:nlos_iid}
\end{figure}

\begin{figure}[tp]
\centering
\includegraphics[scale=0.6]{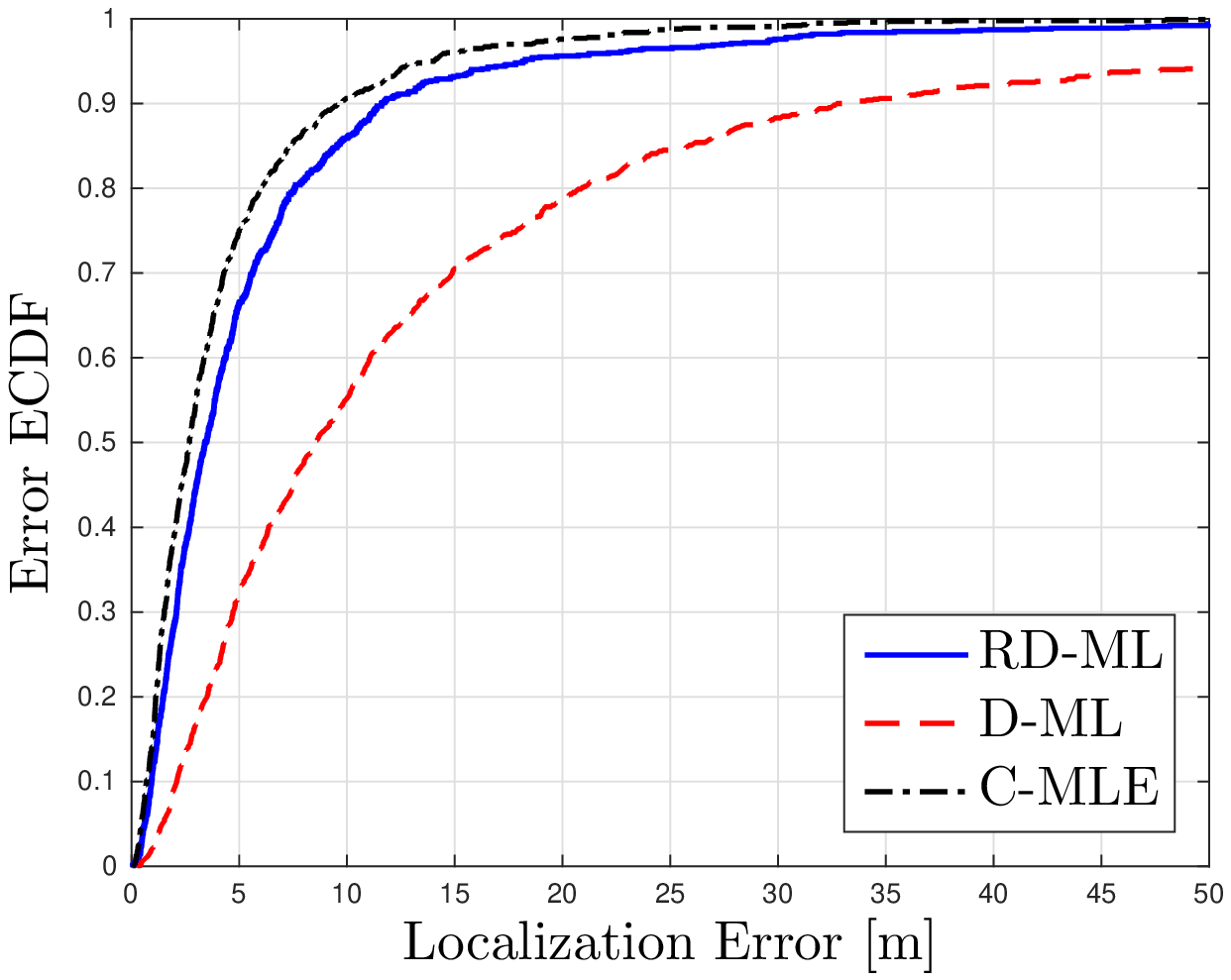}
\caption{Mixed LoS/NLoS scenario with white Gaussian noise. Fully-connected networks, randomized NLoS fraction, $N_a=11, N_u=10$.}
\label{fig:mixed_iid}
\end{figure}

\subsection{Non-Gaussian Noise for NLoS}
In order to evaluate robustness, we consider a model mismatch on the noise distribution by choosing, for the NLoS links, a Student's t-distribution with $\nu = 5$ degrees of freedom (as usual, serially i.i.d. and independent from the LoS noise sequence). For brevity, we consider only the case of a fully-connected network in a mixed LoS/NLoS scenario. As shown by Fig.~\ref{fig:tnoise}, our proposed approach is able to cope with non-Gaussianity without a significant performance loss, thanks to the Gaussian mixture model.

\begin{figure}[tp]
\centering
\includegraphics[scale=0.6]{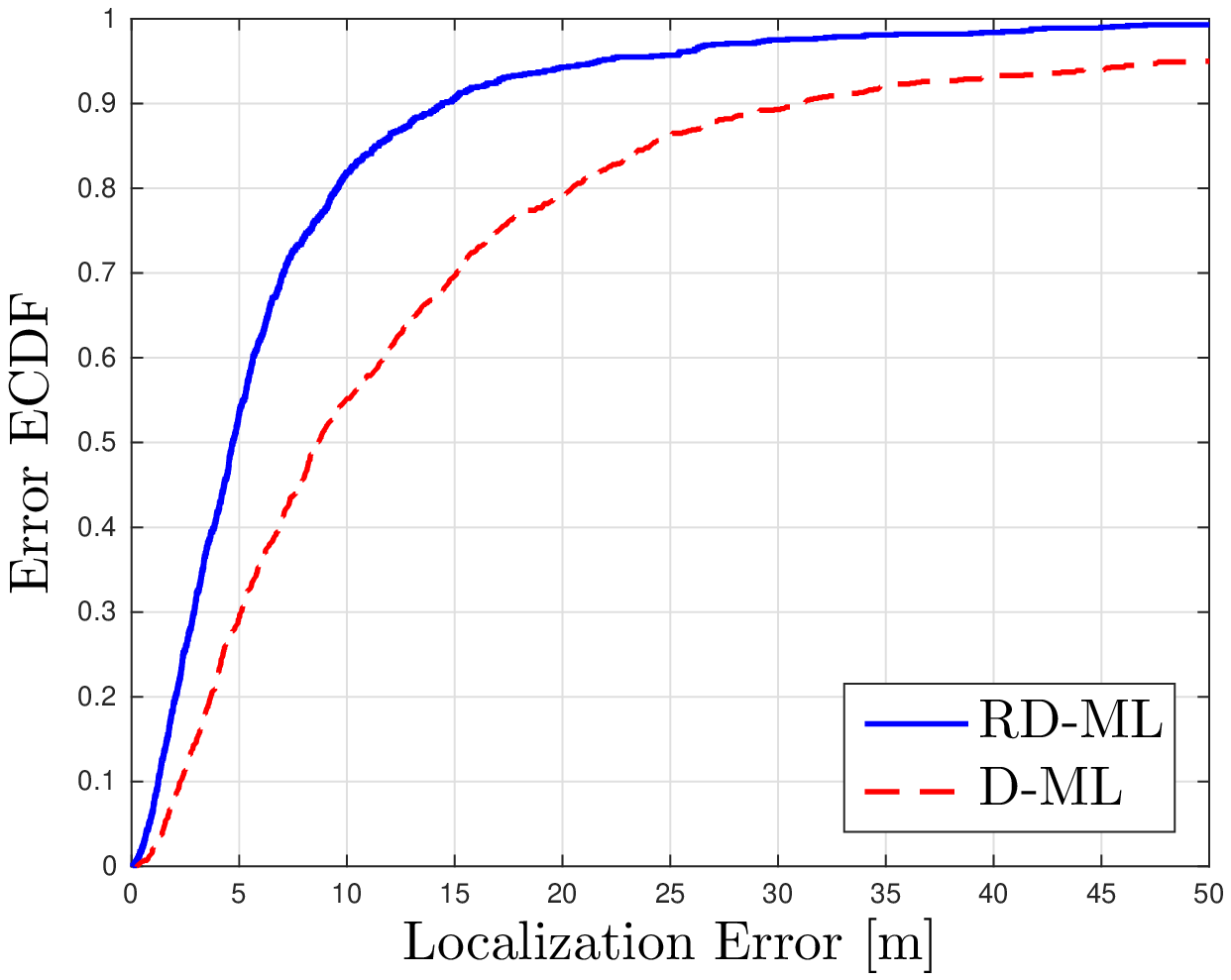}
\caption{Mixed LoS/NLoS scenario with Student's t-distributed noise with $\nu = 5$ under NLoS. Fully-connected networks, randomized NLoS fraction, $N_a=11, N_u=10$.}
\label{fig:tnoise}
\end{figure}

\subsection{Cooperative Gain}
The proposed algorithm exhibits the so-called ``cooperative gain'', i.e., a performance advantage (according to some performance metric, typically localization accuracy) with respect to a non-cooperative approach, where each node tries to localize itself only by self-localization~(\ref{eq:robustsc}). In our case, the cooperative gain is twofold: first, it allows to improve localization accuracy; second, it allows to localize otherwise non-localizable agents. To show the first point, we consider a mixed LoS/NLoS environment (for simplicity, with white Gaussian noise) and a network is generated at each Monte Carlo trial, assuming a communication radius $R = 70 \ \text{m}$, with an ideal probability of detection model~\cite{Handbook}. Moreover, the network is generated in a way as to allow all agents to self-localize and, as a consequence, it is RDML-compatible. In Fig.~\ref{fig:coop}, the ECDFs of the localization error are shown. To show the second point, we consider the toy network depicted in Fig.~\ref{fig:toynet}. In this example, agent $X$ is not able to self-localize, since no anchors are in its neighborhood, $\Gamma_a(X) = \emptyset$. Thus, in a non-cooperative approach, the position of agent $X$ cannot be uniquely determined\footnote{This easily follows by observing that the deterministic (noiseless) version of all the relevant equations for agent $X$ admit infinite solutions.}, while, in a cooperative approach, the position of agent $X$ can actually be obtained by exchanging information, e.g. the estimated positions of its neighbors.

\begin{figure}[tp]
\centering
\includegraphics[scale=0.6]{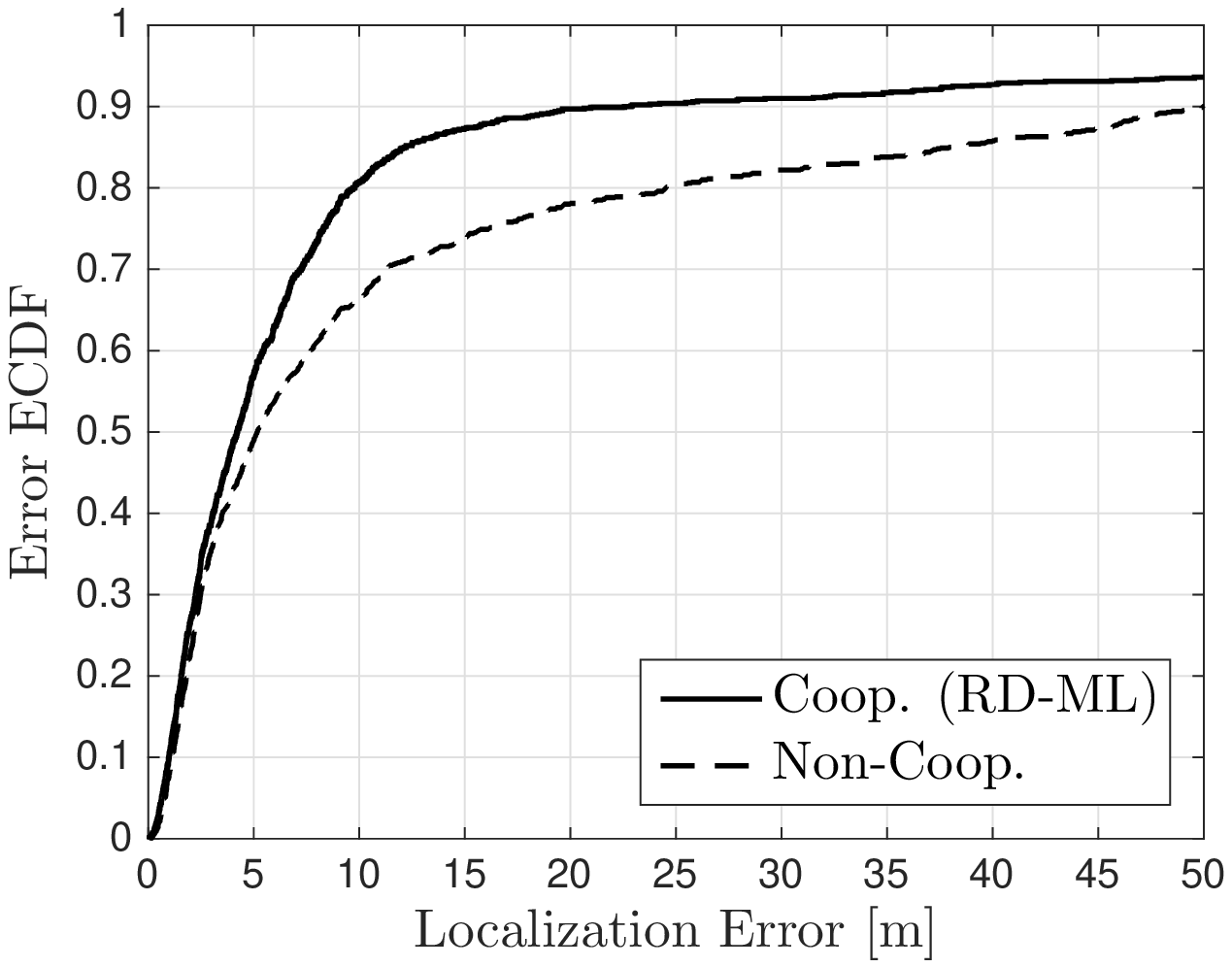}
\caption{Cooperative (RD-ML) and Non-Cooperative approaches in a mixed LoS/NLoS scenario, in a network with communication radius $R = 70 \ \text{m}$, $N_a=11, N_u=10$.}
\label{fig:coop}
\end{figure}

\begin{figure}[tp]
\centering
\includegraphics[scale=0.6]{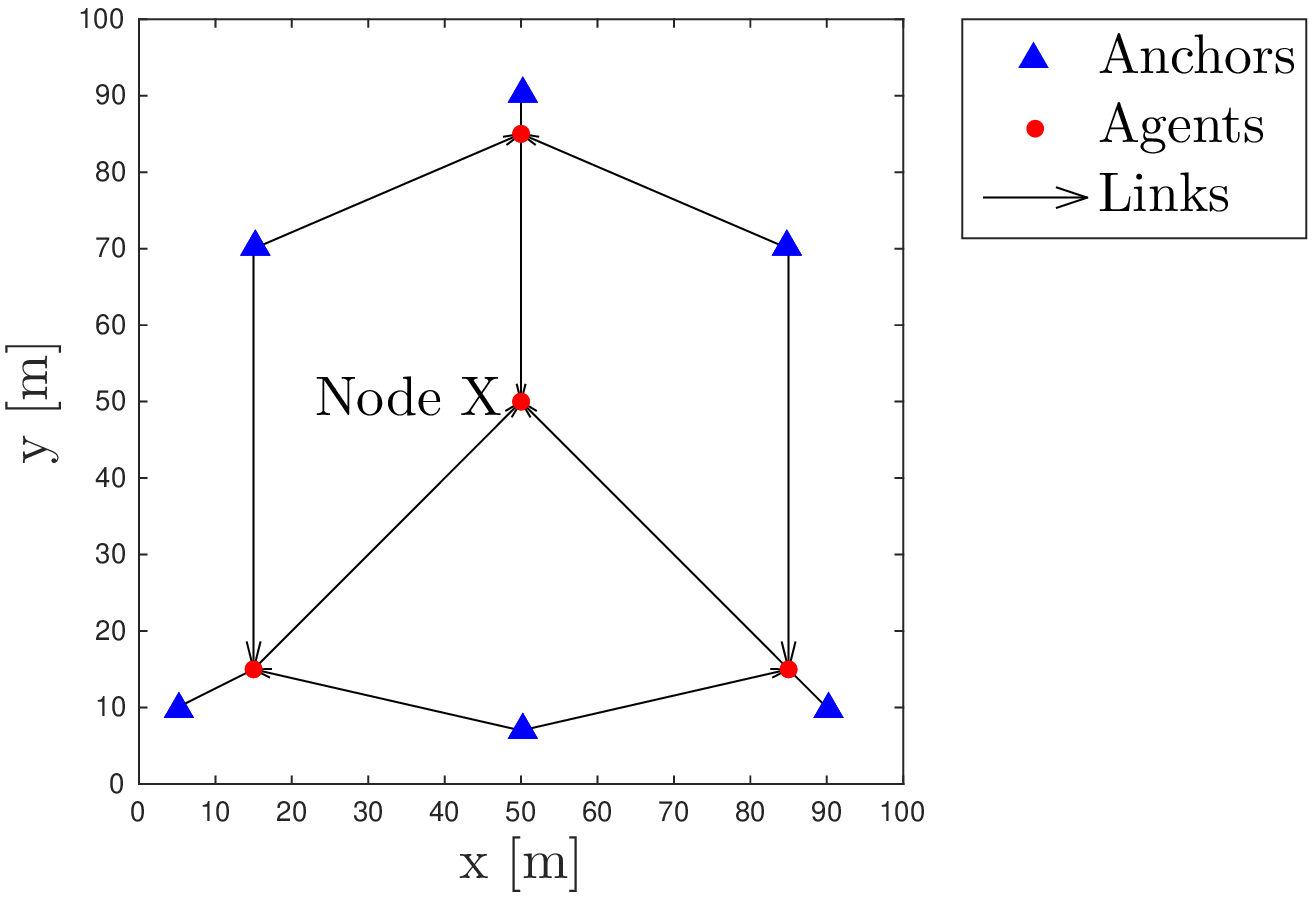}
\caption{Network of a toy example, with $N_a=6, N_u=4$.}
\label{fig:toynet}
\end{figure}

\subsection{Variable $K$ and NLoS fraction}
We next analyze the RD-ML algorithm by varying the number of samples per link, $K$, and the fraction of NLoS links. In both cases, we consider (for brevity) fully-connected networks with white Gaussian noise; for variable $K$, the fraction of NLoS is randomized, while, for variable NLoS fraction, $K$ is fixed. In Fig.~\ref{fig:varK}, we observe that the performance increases as $K$ increases, as expected. In Fig.~\ref{fig:varNLOS}, where each point represents 100 MC trials, the median error is chosen as a performance metric\footnote{This is due to the fact that RSME (Root Mean Square Error) is not a suitable metric when the Error ECDFs are very long-tailed, as in our case.} and RD-ML shows good performances and is not significantly affected by the actual NLoS fraction, which is evidence for its robustness, while D-ML is clearly inferior and suffers from model mismatch.

\begin{figure}[t]
\centering
\includegraphics[scale=0.6]{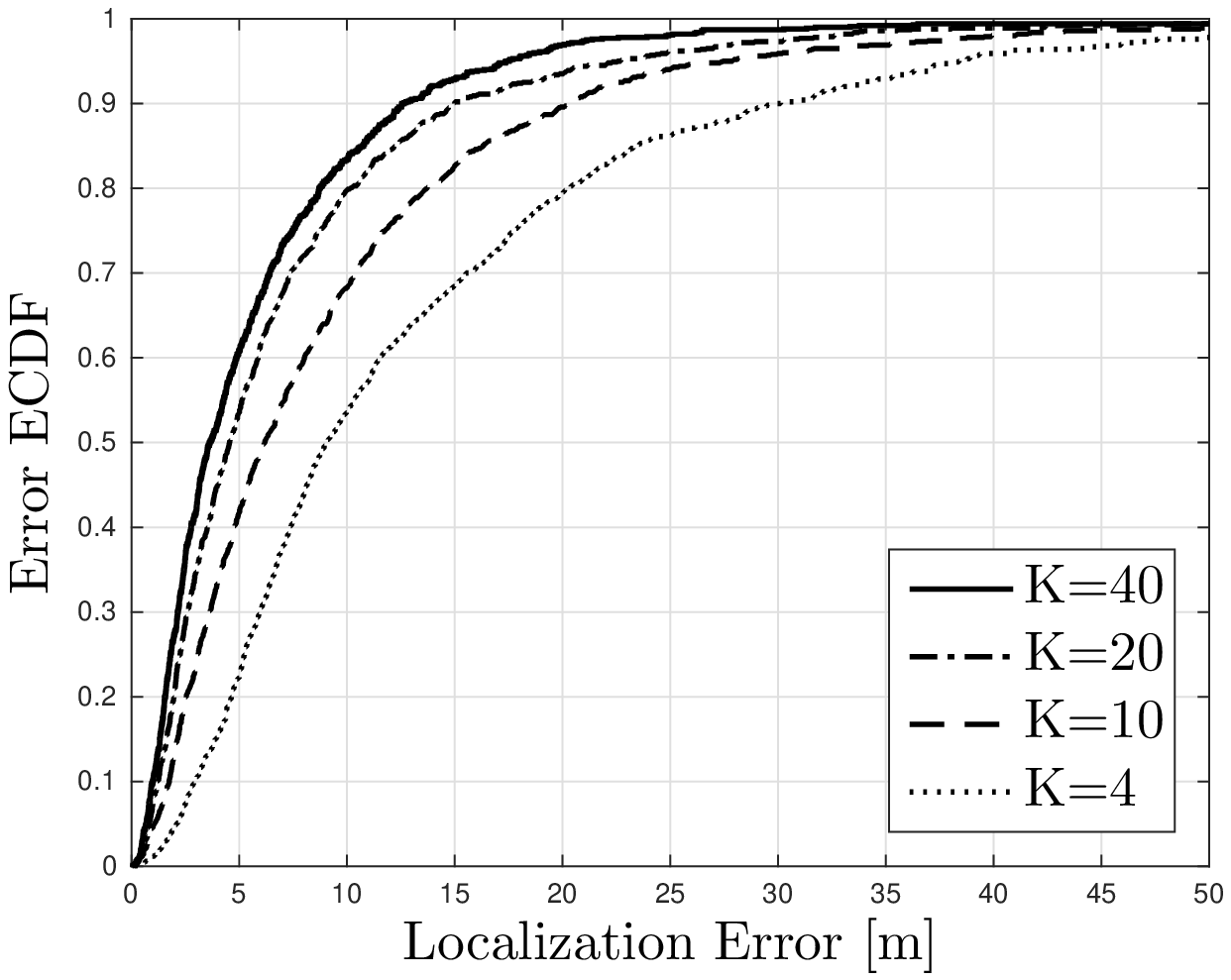}
\caption{Error ECDFs of RD-ML for various values of $K$ in fully-connected networks with white Gaussian noise and randomized NLoS fraction, $N_a=11, N_u=10$.}
\label{fig:varK}
\end{figure}
\begin{figure}[t]
\centering
\includegraphics[trim={15mm 0 0 0},clip, scale=0.63]{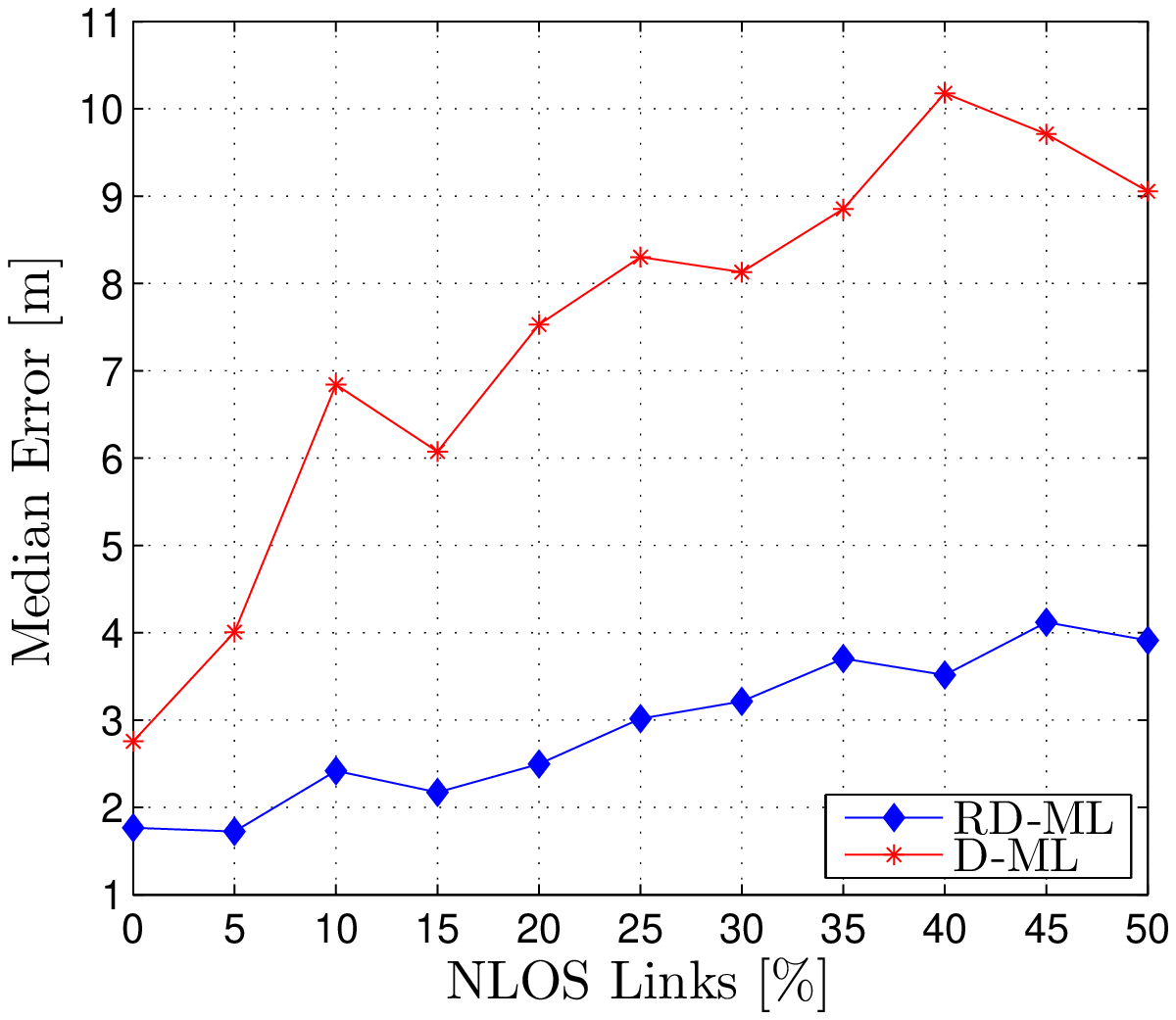}
\caption{Median error of RD-ML and D-ML for a variable fraction of NLoS links. Fully-connected networks, $K=40$, $N_a=11, N_u=10$.}
\label{fig:varNLOS}
\end{figure}

\subsection{Communication and Computational Costs}
The communication overhead is evaluated by computing the number of \emph{elementary messages} sent throughout the network, where an elementary message is simply defined as a single scalar. Sending a $d$-dimensional message, e.g.~$(x,y)$ coordinates, is equivalent to sending $d$ elementary messages. In Fig.~\ref{fig:com_cost}, the communication overhead of RD-ML and D-ML is plotted with respect to the number of agents, assuming a complete graph: for $N_u$ agents, $2 N_u (N_u -1)$ scalars, representing estimated 2D positions, are needed for RD-ML, while $4 N_u (N_u -1)$ are needed for D-ML, as estimates of $(p_0, \alpha)$ are also necessary. Viewing (2D) position has the fundamental information of a message ($d=2$), in a complete graph, RD-ML achieves the theoretical minimum cost in order for all nodes to have complete information, while D-ML uses auxiliary information. For a general graph, the communication overhead depends on the specific graph topology, but is never greater than the aforementioned value. Thus, for a general graph, the communication cost is upper-bounded by a quadratic function of $N_u$.

\begin{figure}[tp]
\centering
\includegraphics[scale=0.6]{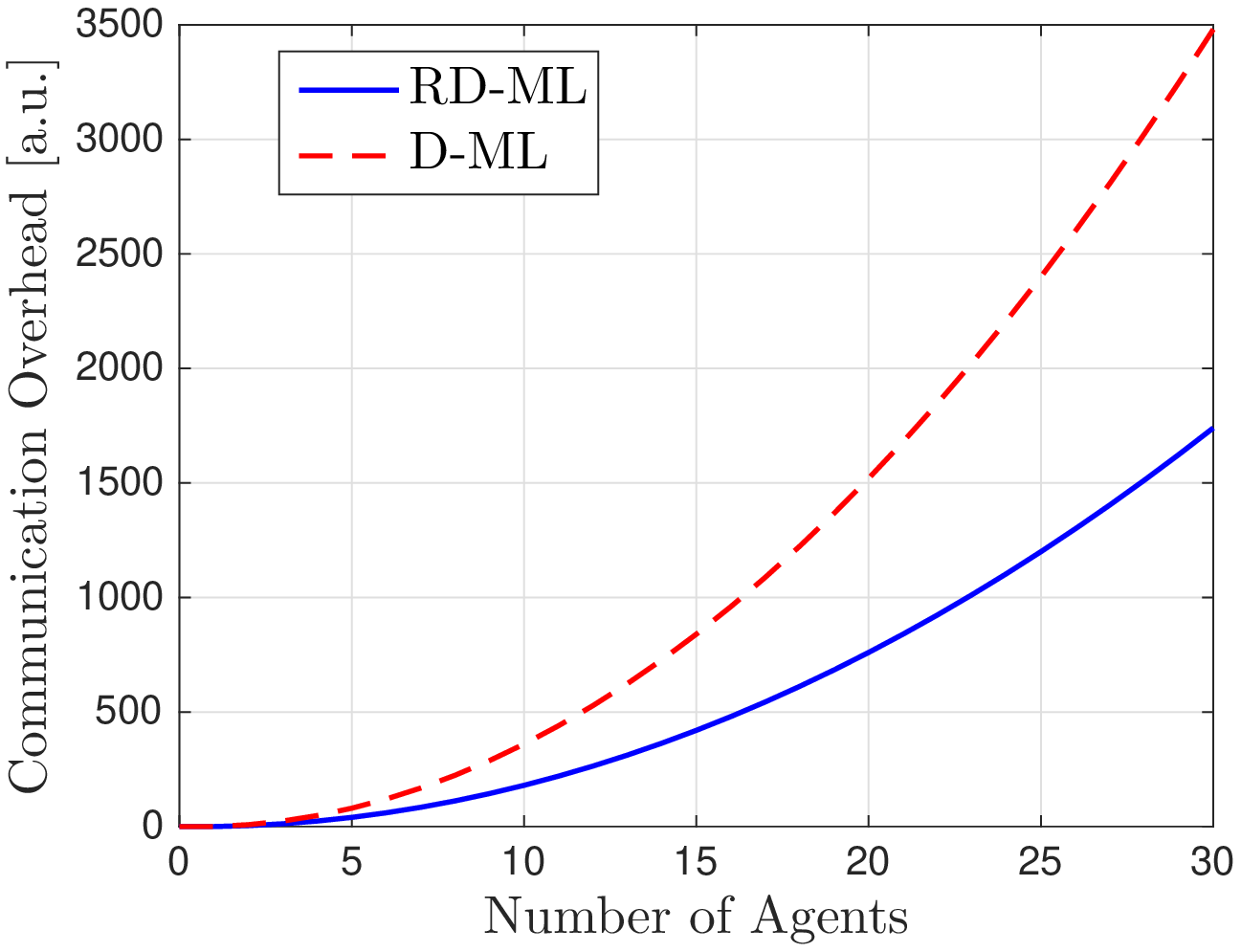}
\caption{Communication costs for a fully-connected network.}
\label{fig:com_cost}
\end{figure}

Regarding computational complexity, both RD-ML and D-ML scale linearly with $N_u$ and they benefit from parallelization, as the main optimization task can be executed independently for each involved node. As already mentioned, D-ML operates a trade-off between communication and computational complexity; in fact, the D-ML optimization problem~(\ref{eq:dmlcost}) is easier to solve than the RD-ML optimization problem~(\ref{eq:robustcost}). Both problems are non-convex and may have local minima/maxima, so care must be taken in the optimization procedure.

\section{Conclusions}
We have developed a novel, robust ML-based scheme for RSS-based distributed cooperative localization in mixed LoS/NLoS scenarios, which, while not being optimal, has good overall accuracy, is adaptive to the environment changes, is robust to NLoS propagation, including non-Gaussian noise, has communication overhead upper-bounded by a quadratic function of the number of agents and computational complexity scaling linearly with the number of agents, also benefiting from parallelization. The main original contributions are that, unlike many algorithms present in the literature, the proposed approach: (a) does not require calibration; (b) does not require symmetrical links (nor does it resort to symmetrization), thus naturally accounting for the topology of directed graphs with asymmetrical links as a result of miss-detections and channel anisotropies. To the best of the authors' knowledge, this is the first distributed cooperative RSS-based algorithm for \emph{directed} graphs. The main disadvantage is imposing some restrictions on the arbitrariness of networks' topology, but these restrictions disappear for sufficiently connected networks. We also derive a compatibility test based on graph coloring, which allows to determine whether the given network is compatible. If it is compatible, convergence of the algorithm is guaranteed. Future work may include the consideration of possible approximations, in order to extend this approach to more general networks, and of alternative models, overcoming the limitations of the standard path loss model.

\begin{appendix}
\section{On Using the Time-Averaged Sample Mean}
\label{sec:A}
Let $\Theta \subseteq \mathbb{R}^d$ be the (non-empty) parameter space, $\bm{\theta} \in \Theta$ unknown deterministic parameters and $s : \Theta \rightarrow \mathbb{R}$ a function. Consider the following model
\begin{equation}
\label{eq:proofmodel}
y(k) = s(\bm{\theta}) + w(k), \quad w(k) \sim \mathcal{N}(0, \sigma^2)
\end{equation}
where: $k=1,\dots,K$ and the noise sequence $w(k)$ is zero-mean, i.i.d.~Gaussian, with deterministic unknown variance $\sigma^2 > 0$. Let
\begin{equation}
\bar{y} = \frac{1}{K} \sum_{k=1}^K y(k)
\end{equation}
be the sample mean and let $p(y(1),\dots,y(k) ; \bm{\theta}, \sigma^2)$ denote the joint likelihood function. Then,
\begin{equation}
\label{eq:A}
\arg \max_{\bm{\theta} \in \Theta} p(y(1),\dots,y(k) ; \bm{\theta}, \sigma^2) = \arg \max_{\bm{\theta} \in \Theta} p(\bar{y} ; \bm{\theta}, \sigma^2)
\end{equation}
\begin{proof}{Since the observations are independent,
\begin{multline}
p(y(1),\dots,y(k) ; \bm{\theta}, \sigma^2) = \prod_{k=1}^K p(y(k) ; \bm{\theta}, \sigma^2) \\
= \frac{1}{(2 \pi \sigma^2)^{K/2}} e^{-\frac{1}{2 \sigma^2} \sum_{k=1}^K ( y(k) - s(\bm{\theta}) ) ^2 }
\end{multline}
We can now focus on the term in the exponential. By adding and subtracting $\bar{y}$ from the quadratic term and expanding, we get
\begin{multline}
% \begin{aligned}
Q = \sum_{k=1}^K ( y(k) - s(\bm{\theta}) ) ^2 = \sum_{k=1}^K ( (y(k) -\bar{y}) - ( s(\bm{\theta}) - \bar{y}) ) ^2 \\
\hspace{3mm}  = \sum_{k=1}^K (y(k) - \bar{y} )^2 + \sum_{k=1}^K ( s(\bm{\theta}) - \bar{y} )^2 - \\
2 \sum_{k=1}^K ( y(k) - \bar{y} )( s(\bm{\theta}) - \bar{y} )
% \end{aligned}
\end{multline}
Observing now that
\begin{multline}
\sum_{k=1}^K ( y(k) - \bar{y} )( s(\bm{\theta}) - \bar{y} ) \\
=  \sum_{k=1}^K ( y(k) s(\bm{\theta}) - y(k) \bar{y} - \bar{y} s(\bm{\theta}) + \bar{y}^2 ) = 0
\end{multline}
we are left with
\begin{equation}
Q = \sum_{k=1}^K (y(k) - \bar{y} )^2 + \sum_{k=1}^K ( s(\bm{\theta}) - \bar{y} )^2
\end{equation}
Thus, the joint pdf can be factorized as follows
\begin{multline}
\hspace{-7mm}
p( y(1),\dots,y(k) ; \bm{\theta}, \sigma^2) = \\ 
\hspace{-12mm}
\underbrace{\frac{1}{(2 \pi \sigma^2)^{K/2}} e^{ -\frac{1}{2 \sigma^2} \sum_{k=1}^K (y(k) - \bar{y} )^2 }}_{ = h(y(1),\dots,y(k) ; \sigma^2)} 
\underbrace{e^{ -\frac{1}{2 \sigma^2} \sum_{k=1}^K ( s(\bm{\theta}) - \bar{y} )^2 }}_{ = g( \bar{y} ; \bm{\theta}, \sigma^2 ) } \hspace{-3mm}
\end{multline}
As a by-product, by invoking the Neyman-Fisher factorization theorem~\cite{Kay1} and assuming known $\sigma^2$, $\bar{y}$ is a sufficient statistic for $\bm{\theta}$. Resuming our proof, we can now observe that
\begin{equation}
\begin{gathered}
\arg \max_{\bm{\theta} \in \Theta} p(y(1),\dots,y(k) ; \bm{\theta}, \sigma^2) = \arg \max_{\bm{\theta} \in \Theta} g( \bar{y} ; \bm{\theta}, \sigma^2 ) \\
= \arg \max_{\bm{\theta} \in \Theta} \left\{ - \sum_{k=1}^K (s(\bm{\theta}) - \bar{y} )^2 \right\} = \arg \min_{\bm{\theta} \in \Theta} (\bar{y} - s(\bm{\theta}) )^2
\end{gathered}
\end{equation}
since no term in the summation depends on $k$. Thus,
\begin{equation}
\arg \max_{\bm{\theta} \in \Theta} p(y(1),\dots,y(k) ; \bm{\theta}, \sigma^2) = \arg \min_{\bm{\theta} \in \Theta} (\bar{y} - s(\bm{\theta}))^2
\end{equation}
On the other hand,
\begin{equation}
\label{eq:appmin}
\begin{aligned}
& \arg \max_{\bm{\theta} \in \Theta} p(\bar{y} ; \bm{\theta}, \sigma^2 ) \\
& = \arg \max_{\bm{\theta} \in \Theta} \frac{1}{\sqrt{2 \pi \frac{\sigma^2}{K}}} \exp \left\{ -\frac{K}{2 \sigma^2} (\bar{y} - s(\bm{\theta}))^2 \right\} \\
&= \arg \max_{\bm{\theta} \in \Theta} \left\{ -\frac{K}{2 \sigma^2} (\bar{y} - s(\bm{\theta}))^2 \right\} \\
&= \arg \min_{\bm{\theta} \in \Theta} (\bar{y} - s(\bm{\theta}) )^2
\end{aligned}
\end{equation}
which completes the proof.}
\end{proof}

\section{C-MLE Derivation}
\label{sec:cmle}
The MLE of $\bm{x}$ is given by
\begin{equation}
\hat{\bm{x}}_{ML} = \arg \max_{\bm{x}} p(\Upsilon ; \bm{x})
\end{equation}
where $p(\Upsilon; \bm{x})$ is the joint likelihood function. Since all other parameters are assumed known, the set of all time-averaged measures $\bar{\Upsilon} = \cup_{i=1}^{N_u} \bar{\mathcal{Y}_i}$ is a sufficient statistic for $\bm{x}$ (see Appendix \ref{sec:A}), from which it follows that
\begin{equation}
\arg \max_{\bm{x}} p(\Upsilon ; \bm{x}) = \arg \max_{\bm{x}} p(\bar{\Upsilon} ; \bm{x})
\end{equation}
By link independence,
\begin{equation}
p( \bar{\Upsilon} ; \bm{x}) = \prod_{i=1}^{N_u} \prod_{j \in \Gamma(i)}
p( \bar{r}_{j \rightarrow i} ; \bm{x})
\end{equation}
where $p(\bar{r}_{j \rightarrow i} ; \bm{x})$ is the marginal likelihood function. Thus, we have
\begin{multline}
\hat{\bm{x}}_{ML} = \arg \max_{\bm{x}} \prod_{i=1}^{N_u} \prod_{j \in \Gamma(i)} \frac{1}{(2 \pi \sigma^2_{j \rightarrow i})} \times \\
e^{ -\frac{1}{2 \sigma^2_{j \rightarrow i}} ( \bar{r}_{j \rightarrow i} -p_{0_{j\rightarrow i}} +10 \alpha_{j \rightarrow i} \log_{10} \| \bm{x}_j - \bm{x}_i  \| )^2 }
\end{multline}
Taking the natural logarithm and neglecting constants,
\begin{multline}
\hat{\bm{x}}_{ML} = \arg \min_{\bm{x}} \sum_{i=1}^{N_u} \sum_{j \in \Gamma(i)} \frac{1}{\sigma^2_{j \rightarrow i}} ( \bar{r}_{j \rightarrow i} -p_{0_{j\rightarrow i}} +\\
10 \alpha_{j \rightarrow i} \log_{10} \| \bm{x}_j - \bm{x}_i  \| )^2
\end{multline}
which completes the derivation.
\end{appendix}

%%%%%%%%%%%%%%%%%%%%%%%%%%%%%%%%%%%%%%%%%%%%%%
%%                                          %%
%% Backmatter begins here                   %%
%%                                          %%
%%%%%%%%%%%%%%%%%%%%%%%%%%%%%%%%%%%%%%%%%%%%%%

\begin{backmatter}

\section*{List of Abbreviations}
\begin{abbrv}
\item[BP]					 Belief Propagation
\item[C-MLE]				 Centralized Maximum likelihood Estimator
\item[D-ECM]				 Distributed Expectation-Conditional Maximization
\item[D-ML]					 Distributed Maximum Likelihood
\item[ECDF]					 Empirical Cumulative Distribution Function
\item[ECM]					 Expectation-Conditional Maximization
\item[EM]					 Expectation-Maximization
\item[LoS]                   Line-of-Sight
\item[LS]					 Least Squares
\item[MC]					 Monte Carlo
\item[ML]					 Maximum Likelihood
\item[MLE]                   Maximum Likelihood Estimator
\item[NBP]					 Nonparametric Belief Propagation
\item[NLoS]					 Non-Line-of-Sight
\item[RD-ML]				 Robust Distributed Maximum Likelihood
\item[RF]					 Radiofrequency
\item[RMSE]                  Root Mean Square Error
\item[RSS]					 Received Signal Strength
\item[RSSI]					 Received Signal Strength Indicator
\item[SDP]					 Semidefinite-Programming
\item[SPAWN]				 Sum-Product Algorithm over Wireless Networks
\item[WLS]					 Weighted Least Squares
\item[WSN]					 Wireless Sensor Network
\item[2D]					 2-dimensional
\end{abbrv}

\section*{Declarations}

\section*{Availability of data and material}
The datasets used and/or analyzed during the current study are available from the corresponding author on reasonable request.

\section*{Competing interests}
The authors declare that they have no competing interests.

\section*{Funding}
The work of L. Carlino was supported by the "Erasmus+ Traineeship" programme of University of Salento. 
The work of M. Muma was supported by the "Athene Young Investigator Programme" of Technische Universität Darmstadt.

\section*{Authors contributions}
LC has contributed towards the development of the proposed algorithms and the performance analysis. DJ has contributed towards the introduction, the related work and with other minor revisions throughout the paper. MM has contributed towards the example networks regarding graph connectivity, the overall organization of the paper, and with other minor revisions throughout the paper. As the supervisor, AMZ has proofread the paper several times and provided guidance throughout the whole preparation of the manuscript. All authors read and approved the final manuscript.

\section*{Acknowledgements}
The authors would like to thank prof.~F.~Bandiera from University of Salento for having started the collaboration which lead to this work.
%The authors would like to thank the reviewers for their thorough reviews and helpful suggestions.

\section*{Authors information}
\textbf{LC} received his B.Sc. degree with full marks in Information Engineering at University of Salento, Lecce, Italy, in 2012. He received his M.Sc. degree in Telecommunications Engineering (\emph{summa cum laude}) at University of Salento, in 2014. He started his Ph.D. studies at University of Salento in 2014 and finished his Ph.D. in 2018. His main area of research is signal processing, while his research topic is applying signal processing techniques to the problem of localization, with a specific focus on localization based upon received signal strength. His other research interests include radar applications, target tracking, data mining and information theory.

\textbf{DJ} received the B.Sc. degree in information and communication engineering from Zhejiang University, Hangzhou, China, in 2011, and the M.Sc. degree in electrical engineering and information technology from Technische Universit$\rm{\ddot{a}}$t Darmstadt, Darmstadt, Germany, in 2014. She is currently working towards the Ph.D. degree in the Signal Processing Group at Technische Universit$\rm{\ddot{a}}$t Darmstadt. Her research interests include localization and tracking, distributed and cooperative inference in wireless networks.

\textbf{MM} received the Dipl.-Ing (2009) and the Dr.-Ing. degree (\emph{summa cum laude}) in Electrical Engineering and Information Technology (2014), both from Technische Universit\"{a}t Darmstadt, Darmstadt, Germany. He completed his diploma thesis with the Contact Lens and Visual Optics Laboratory, School of Optometry, Brisbane, Australia, on the role of cardiopulmonary signals in the dynamics of the eye's wavefront aberrations. Currently, he is a Postdoctoral Fellow at the Signal Processing Group, Institute of Telecommunications and has recently been awarded Athene Young Investigator of Technische Universit\"{a}t Darmstadt. His research is on robust statistics for signal processing with applications in biomedical signal processing, wireless sensor networks, and array signal processing. MM was the supervisor of the Technische Universit{\"a}t Darmstadt student team who won the international IEEE Signal Processing Cup 2015. MM co-organized the 2016 Joint IEEE SPS and EURASIP Summer School on Robust Signal Processing. In 2017, together with his coauthors, MM received the IEEE Signal Processing Magazine Best Paper Award for the paper entitled ``Robust Estimation in Signal Processing: A tutorial-style treatment of fundamental concepts".  In 2017, he was elected to the European Association for Signal Processing (EURASIP) Special Area Team in Theoretical and Methodological Trends in Signal Processing (SAT-TMSP).

\textbf{AZ} is a Fellow of the IEEE and IEEE Distinguished Lecturer (Class 2010- 2011). He received his Dr.-Ing. from Ruhr-Universit\"at Bochum, Germany in 1992. He was with Queensland University of Technology, Australia from 1992-1998 where he was Associate Professor. In 1999, he joined Curtin University of Technology, Australia as a Professor of Telecommunications. In 2003, he moved to Technische Universit\"at Darmstadt, Germany as Professor of Signal Processing and Head of the Signal Processing Group. His research interest lies in statistical methods for signal processing with emphasis on bootstrap techniques, robust detection and estimation and array processing applied to telecommunications, radar, sonar, automotive monitoring and safety, and biomedicine. He published over 400 journal and conference papers on the above areas. AZ served as General Chair and Technical Chair of numerous international conferences and workshops; more recently he was the Technical Co-Chair of ICASSP-14 held in Florence, Italy. He also served on publication boards of various journals, notably as Editor-In-Chief of the IEEE Signal Processing Magazine (2012-2014). AZ was the Chair (2010-2011) of the IEEE Signal Processing Society (SPS) Technical Committee Signal Processing Theory and Methods (SPTM). He served on the Board of Governors of the IEEE SPS (2015-2017) and is the president of the European Association of Signal Processing (EURASIP) (2017-2018).

%%  After submission of the .TEX file,                     %%
%%  you will be prompted to submit your .BBL file.         %%

% if your bibliography is in bibtex format, use those commands:
%\bibliographystyle{bmc-mathphys} % Style BST file (bmc-mathphys, vancouver, spbasic).
%\bibliography{biblio}      % Bibliography file (usually '*.bib' )
% for author-year bibliography (bmc-mathphys or spbasic)
% a) write to bib file (bmc-mathphys only)
% @settings{label, options="nameyear"}

% or include bibliography directly:
\bibliographystyle{bmc-mathphys}
\bibliography{biblio}

\end{backmatter}
\end{document}